\documentclass[sigplan,10pt]{acmart}

  \acmConference{Arxiv version}{November 2019}{}
  \acmYear{2019}
  \acmISBN{} 
  \acmDOI{} 
  \startPage{}

\setcopyright{none}

\usepackage{booktabs}   
\usepackage{subcaption} 

\usepackage{amsmath}
\usepackage{amssymb}
\usepackage{amsthm}
\usepackage{graphicx}
\usepackage[colorinlistoftodos]{todonotes}
\usepackage{algorithm}
\usepackage[noend]{algpseudocode}
\usepackage{enumerate}
\usepackage{xcolor}
\usepackage{framed}
\usepackage{bm}
\usepackage{pifont}
\usepackage{multicol}
\usepackage{lipsum}
\usepackage{tikz}
\usetikzlibrary{calc}
\usepackage{array}
\usepackage{relsize}
\usepackage{comment}
\usepackage{url}
\usepackage[title]{appendix}
\usepackage{caption}
\usepackage{subcaption}
\usepackage{xstring}
\usepackage{times}
\usepackage{wrapfig}
\usepackage{varwidth}
\usepackage{color}
\usepackage[leftcaption]{sidecap}

\usepackage[english]{babel}
\usepackage[utf8x]{inputenc}
\usepackage[T1]{fontenc}

\usepackage{xspace}      

\newtheorem{observation}{Observation}
\newtheorem{claim}{Claim}
\newtheorem{lemma}{Lemma}
\newtheorem{corollary}{Corollary}

\newtheorem{definition}{Definition}

\algdef{SE}[RECEIVING]{Receiving}{EndReceiving}[1]{\textbf{upon
receiving}\ #1\ \algorithmicdo}{\algorithmicend\ \textbf{}}%
\algtext*{EndReceiving}

\algdef{SE}[UPON]{Upon}{EndUpon}[1]{\textbf{upon
}\ #1\ \algorithmicdo}{\algorithmicend\ \textbf{}}%
\algtext*{EndUpon}

\newcommand{\lr}[1]{\langle #1 \rangle}

\newcommand{\DM}[1]{\textbf{DM:} \textcolor{purple}{#1}}
\newcommand{\arik}[1]{\textbf{arik:} \textcolor{green}{#1}}

\newcommand{\ignore}[1]{}

\newcommand{\propose}{$\mathsf{engage}$\xspace}
\newcommand{\vc}{$\mathsf{wedge\&exchange}$\xspace}
\newcommand{\barriersync}{$\mathsf{barrier\text{-}sync}$\xspace}
\newcommand{\barrierready}{$\mathsf{barrier\text{-}ready}$\xspace}
\newcommand{\elect}{$\mathsf{elect}$\xspace}

\def\HiLi{\leavevmode\rlap{\hbox to
\hsize{\color{gray!10}\leaders\hrule height .8\baselineskip
depth .5ex\hfill}}}

\algblockdefx[ON]{On}{EndOn}[1]
  {{\bf On} #1\   \algorithmicdo }
  {\algorithmicend\ }
\makeatletter
\ifthenelse{\equal{\ALG@noend}{t}}%
  {\algtext*{EndOn}}
  {}%
\makeatother

\algdef{SE}[RECEIVING]{Receiving}{EndReceiving}[1]
	{\textbf{upon receiving}\ #1\ \algorithmicdo}

\makeatletter

\tikzset{%
  remember picture with id/.style={%
    remember picture,
    overlay,
    save picture id=#1,
  },
  save picture id/.code={%
    \edef\pgf@temp{#1}%
    \immediate\write\pgfutil@auxout{%
      \noexpand\savepointas{\pgf@temp}{\pgfpictureid}}%
  },
  if picture id/.code args={#1#2#3}{%
    \@ifundefined{save@pt@#1}{%
      \pgfkeysalso{#3}%
    }{
      \pgfkeysalso{#2}%
    }
  }
}

\def\savepointas#1#2{%
  \expandafter\gdef\csname save@pt@#1\endcsname{#2}%
}

\def\tmk@labeldef#1,#2\@nil{%
  \def\tmk@label{#1}%
  \def\tmk@def{#2}%
}

\tikzdeclarecoordinatesystem{pic}{%
  \pgfutil@in@,{#1}%
  \ifpgfutil@in@%
    \tmk@labeldef#1\@nil
  \else
    \tmk@labeldef#1,(0pt,0pt)\@nil
  \fi
  \@ifundefined{save@pt@\tmk@label}{%
    \tikz@scan@one@point\pgfutil@firstofone\tmk@def
  }{%
  \pgfsys@getposition{\csname save@pt@\tmk@label\endcsname}\save@orig@pic%
  \pgfsys@getposition{\pgfpictureid}\save@this@pic%
  \pgf@process{\pgfpointorigin\save@this@pic}%
  \pgf@xa=\pgf@x
  \pgf@ya=\pgf@y
  \pgf@process{\pgfpointorigin\save@orig@pic}%
  \advance\pgf@x by -\pgf@xa
  \advance\pgf@y by -\pgf@ya
  }%
}

\newlength\AlgIndent
\newcounter{mymark}
\newcommand\ColorLine{%
  \stepcounter{mymark}%
  \tikz[remember picture with id=mark-\themymark,overlay] {;}%
  \begin{tikzpicture}[remember picture,overlay]%
    \filldraw[gray!20]%
   let \p1=(pic cs:mark-\themymark), 
   \p2=(current page.east)  in 
   ([xshift=-\ALG@thistlm-0em,yshift=-0.7ex]0,\y1)  rectangle
   ++(\linewidth+\AlgIndent,\baselineskip); \end{tikzpicture}%
}%

\newcommand\ColorLinex{%
  \stepcounter{mymark}%
  \tikz[remember picture with id=mark-\themymark,overlay] {;}%
  \begin{tikzpicture}[remember picture,overlay]%
    \filldraw[gray!20]%
   let \p1=(pic cs:mark-\themymark), 
   \p2=(current page.east)  in 
   ([xshift=-\ALG@thistlm--3em,yshift=-0.7ex]0,\y1) 
   rectangle ++(\linewidth+\AlgIndent,\baselineskip); \end{tikzpicture}%
}%

\algnewcommand\CREQUIRE{\item[\setlength\AlgIndent{1.6em}\ColorLine\algorithmicrequire]}%
\algnewcommand\CENSURE{\item[\setlength\AlgIndent{1.6em}\ColorLine\algorithmicensure]}%
\algnewcommand\CSTATE{\State\ColorLine}%
\algnewcommand\CSTATEx{\Statex\ColorLinex}%
\algnewcommand\CCOMMENT{\Comment\ColorLine}%

   \algdef{SE}[IF]{CIF}{ENDIF}%
   [2][default]{\ColorLine\algorithmicif\ #2\ \algorithmicthen}%
   {\algorithmicend\ \algorithmicif}%
   
   \algdef{SE}[FOR]{CFOR}{ENDFOR}%
   [2][default]{\ColorLine\algorithmicfor\ #2\ \algorithmicdo}%
   {\algorithmicend\ \algorithmicfor}%
\makeatother

\begin{document}

\title[ACE: Abstract Consensus Encapsulation]{ACE: Abstract Consensus
Encapsulation\\ for Liveness Boosting of State Machine Replication}   


\author{Alexander Spiegelman}
\affiliation{
  \institution{VMware Research}            
}

\author{Arik Rinberg}
\affiliation{
  \institution{Technion}           
}

\begin{abstract}

  With the emergence of cross-organization attack-prone byzantine
  fault-tolerant (BFT) systems, so-called Blockchains, providing
  asynchronous state machine replication (SMR) solutions is no longer a theoretical concern.
  This paper introduces \emph{ACE}: a general framework for the
  software design of fault-tolerant SMR systems.
  We first propose a new \emph{leader-based-view (LBV)} abstraction
  that encapsulates the core properties provided by each view in a
  partially synchronous consensus algorithm, designed according to the
  leader-based view-by-view paradigm (e.g., PBFT and Paxos).
  Then, we compose several LBV instances in a non-trivial way
  in order to boost asynchronous liveness of existing SMR solutions.
    
  ACE is model agnostic -- it abstracts away any model assumptions
  that consensus protocols may have, e.g., the ratio
  and types of faulty parties.
  For example, when the LBV abstraction is instantiated with a
  partially synchronous consensus algorithm designed to tolerate
  crash failures, e.g., Paxos or Raft, ACE yields an asynchronous SMR
  for $n = 2f+1$ parties.
  However, if the LBV abstraction is instantiated with a byzantine
  protocol like PBFT or HotStuff, then ACE yields an asynchronous byzantine SMR for
  $n = 3f+1$ parties.
  
  To demonstrate the power of ACE, we implement it
  in C++, instantiate the LBV abstraction with a view implementation
  of HotStuff -- a state of the art partially synchronous byzantine
  agreement protocol -- and compare it with the base HotStuff
  implementation under different adversarial scenarios.
  Our evaluation shows that while ACE is outperformed by
  HotStuff in the optimistic, synchronous, failure-free case, ACE
  has absolute superiority during network asynchrony and
  attacks.

  \end{abstract}



\maketitle

\section{Introduction}
\label{sec:intro}

In practice, building reliable systems via state machine replication (SMR)
requires resilience against all network conditions, including malicious attacks.
The best way to model such settings is by assuming asynchronous
communication links. However, due to the FLP
result~\cite{fischer1982impossibility}, deterministic asynchronous SMR
solutions are impossible.

Two principal approaches are used to circumvent
this result.
The first is by assuming \emph{partial
synchrony}~\cite{dwork1988consensus}, in which protocols are designed
to guarantee safety under worst case network conditions, but are able
to satisfy progress only during ``long enough'' periods of network
synchrony.
The vast majority of the protocols in this model follow the
leader-based view-by-view paradigm due to their speed during
synchronous attack-free network periods, and their relative
simplicity.
In fact, most deployed systems, several of which have become the de
facto standards for building reliable systems (e.g.,
Paxos~\cite{lamport2001paxos}, PBFT~\cite{pbft} and others~\cite{
zookeeper, ongaro2014search, etcd, hotstuff, bessani2014state}),
adopt this approach.
The drawback of the partial synchrony model is fact that it fails to
capture mobile networks attacks~\cite{santoro1989time}, leaving the
leader-based view-by-view algorithms vulnerable.
For example, an attacker can prevent progress by adaptively blocking
the communication of one party (the leader of the current view) at a
time.


The second approach to circumvents the FLP impossibility is by
employing randomization~\cite{ben1983another, rabin1983randomized}.
The most commonly used strategy is to always satisfy safety
properties, but relaxing liveness to guarantee eventual
progress with probability approaching $1$ under all network
conditions.
Potentially, protocols designed for the asynchronous communication
model can operate at network speed, but unfortunately, they are
rarely deployed in practice due to their complexity and overhead, and
are mostly the focus of theoretical academic work.



\paragraph{Main contribution.} 
In this paper we combine the best of both approaches.
We present \emph{ACE}, a simple
generic framework for \emph{asynchronous boosting}, which converts
consensus (also called agreement) algorithms designed according the
\emph{leader-based view-by-view paradigm} in the partial synchrony model into randomized fully asynchronous SMR solutions.
ACE is model agnostic -- it has no model assumptions, and thus can be
applied to any leader-based protocol in the byzantine or crash
failure model.
As a result, with ACE, a system designer can benefit twofold. 
On the one hand, from the experience gained in decades of
leader-based view-by-view algorithm design and engineered systems,
and on the other hand, from a robust asynchronous solution.





\paragraph{View-by-view paradigm.}

ACE is general and applicable
to a family of consensus leader-based view-by-view protocols
designed for the partially synchronous communication
model~\cite{lamport2001paxos, pbft, hotstuff}.
Such protocols divide executions into a sequence of views, each with
a designated leader.
Every view is then further divided into two phases. 
First, the \emph{leader-based} phase in which the designated leader
tries to drive progress by getting all parties to commit its
value.
Then, when parties suspect that the leader is faulty, whether it is
really faulty or due to asynchrony or network attacks, they start the
\emph{view-change} phase in which they \emph{exchange} information to
safely \emph{wedge} the current view 
and move to the next one.

%

\paragraph{Technical contribution.}
ACE's first contribution is providing a formal
characterization of the leader-based view-by-view protocols by
defining a \emph{leader-based view (LBV)} abstraction, which
encapsulates the core properties of a single view and provides an API
that allows de-coupling of the \emph{leader-based} phase from the
\emph{view-change} phase.
%
%
%
In the view-by-view paradigm, view-change phases are triggered
with timers: parties start a timer at the beginning of the
leader-based phase in each view and if the timer expires before the
leader drives progress, parties move to the view-change phase.
Indeed, if we instantiate the LBV abstraction with an implementation
of a view of some view-by-view protocol and operate a sequence of
these LBVs, each time invoking the leader-based phase, then timeout, and then
invoke the view-change phase, then we end up with a variant of the
view-by-view protocol that the LBV is instantiated with.
See illustration in Figure~\ref{fig:recProtocol}.

\begin{figure}[th]
    \begin{center}
        \includegraphics[width=0.48\textwidth]{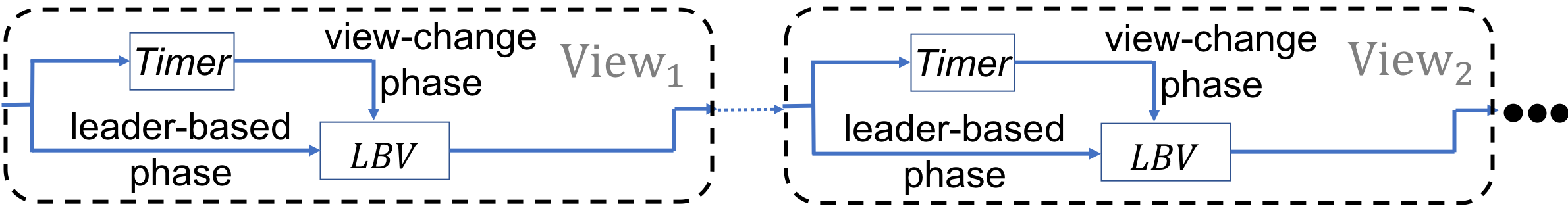}
    \end{center}
    \caption{Using a sequence of LBV instances to reconstruct a
    partially synchronous leader-based view-by-view protocol.}
    \label{fig:recProtocol}
\end{figure} 

ACE's second technical contribution is a novel \emph{wave} mechanism
to control the trigger of view-change phases.
Rather than using timers, the wave mechanism uses the API provided by
the LBV abstraction to generically rearrange views in view-by-view
protocols.
The key idea is to compose several LBV abstractions in a way that
allows progress at network speed during periods of asynchrony.
%
%
%
This mechanism exploits a key property shared by all view-by-view
protocols: If the leader of a view is correct and timers never
expire, then eventually a decision will be made in this
view.

ACE's single-shot agreement protocol proceeds in a
\emph{wave-by-wave} manner, each wave operates as follows:
Instead of running one LBV instance (as view-by-view protocols do),
a wave runs $n$ LBV instances (the leader-based phase) simultaneously,
each with a distinct designated leader.
Then, the wave performs a barrier synchronization in which
parties wait until a quorum of the instances have reached a decision.
Recall that by the key property, all instances with correct leaders
eventually decide so eventually the barrier is reached.

After the barrier is reached, one LBV instance is selected
unpredictably and uniformly at random.
The chosen instance ``wins'', and all other instances are ignored.
Then, parties use the LBV's API to invoke the view-change phase in the
chosen instance (only).
The view-change phase here has two purposes.
First, it \emph{boosts termination}. If the chosen instance has
reached a decision, meaning that a significantly large quorum of
parties decided, then all correct parties learn this decision during
the view-change phase.
Second, as in every view-by-view protocol, the view-change
phase \emph{ensures safety} by forcing the leaders of the next wave
to propose safe values.


The next wave enacts $n$ new LBV instances, each with a different
leader that proposes a value according to the state returned
from the view-change phase of the chosen instance of the previous
wave.
Note that since parties wait for a large quorum of LBV instances to
reach a decision in each wave before randomly choosing one, the chosen
LBV has a constant probability of having a decision, hence,
together with the termination boosting provided by the view-change
phase, we get progress, in expectation, in a constant number of
waves.

As to SMR, ACE implements a variant in which parties do not proceed
to the next slot before they learn the decision value of the current
one, but once they move to the next one they stop participating in
the current slot and garbage collect all the associated resources.
Deferring next slots until the current decision is known is essential
for systems in which the validity of a value for a certain slot
depends on all previous decision values (e.g., Blockchains).
ACE's SMR solution uses an instance of our single-shot protocol
for every slot together with a forwarding mechanism to help slow
parties catch-up.

\paragraph{Applications.}
ACE can take any view-by-view consensus protocol designed for the
partially synchronous model and transform it into an asynchronous SMR
solution.
%
In order to instantiate ACE with a specific algorithm,
e.g., PBFT~\cite{pbft} or Paxos~\cite{lamport2001paxos}, one must
take the core of the algorithm logic (a single view) and wrap it with
the LBV's API that provides an \propose method to start the
leader-based phase and a \vc method to proceed to the view-change
phase.
We define the properties each LBV implementation has to satisfy, and
argue that all existing leader-based view-by-view algorithms
implicitly satisfy these properties.
Therefore, instantiating ACE does not require new logic
implementation beyond the engineering effort of providing its API.
Furthermore, ACE's modularity provides a clean separation of
concerns between the Safety (provided by the LBV properties) and
asynchronous Liveness (provided by the framework).
With ACE, designing a new fully asynchronous SMR system based on an
existing partially synchronous consensus protocol or updating a
deployed system with a novel agreement protocol, only requires an
LBV implementation.

An important feature of ACE is fairness. 
Due to the inherent unpredictable randomness in the election
mechanism, ACE provides an equal chance for all correct parties to
drive decisions during synchrony, and optimally bounds the
probability to decide a value proposed by a correct party during
asynchrony.
Another important feature of ACE is its model agnostic,
namely it does not add any model assumptions on top of the
assumptions made by the instantiated protocols.
As a result, when instantiated with a BFT protocol such as
PBFT~\cite{pbft} we get an asynchronous byzantine state machine
replication, and when instantiated with a crash-failure solution like
Paxos~\cite{lamport2001paxos} or Raft~\cite{ongaro2014search} we get
the first asynchronous SMR system tolerating any minority of
failures.

In order to make ACE model agnostic, we encapsulate the
barrier synchronization and the random election mechanisms into
abstractions as well, denoted \emph{barrier} and
\emph{leader-election} respectively, and define their properties.
So that when instantiating our framework, one need also provide
implementation of these abstractions for the assumed model.
In Section~\ref{sec:evaluation} we give an implementation
example for the byzantine model.

\paragraph{Complexity and performance.}

To demonstrate ACE, we choose to focus on the byzantine
model as this is the model considered by Blockchain systems and
we believe that, due to their high stakes, Blockchain systems will
benefit the most from a generic asynchronous SMR solution that can
tolerate network attacks.
We implement ACE's algorithms in C++ and instantiate the LBV
abstraction with a variant of HotStuff~\cite{hotstuff} -- a state of
the art BFT solution, which is currently being implemented in several
commercial Blockchain systems~\cite{libra}.
We compare the ACE instantiation to the base (raw) HotStuff
implementation in different scenarios. Our evaluation shows that
while base HotStuff outperforms ACE (instantiated with HotStuff) in
the optimistic, synchronous, failure-free case, ACE has absolute
superiority during asynchronous periods and network attacks. 
For example, we show that byzantine parties can hinder progress in
base HotStuff by targeting leaders with a DDoS attack, whereas ACE
manages to commit values at network speed.

From a theoretical perspective, ACE generalizes the idea in
VABA~\cite{vaba}, the first asymptotically optimal asynchronous byzantine agreement
protocol, and adds only factor $n$ in communication complexity over
the leader-based phase in each view of the protocol it is instantiated
with.
Therefore, since no protocol can solve byzantine agreement with less
than quadratic communication~\cite{dolev1983authenticated},
ACE can leverage protocols with linear communication
in the leader-based phases, like HotStuff~\cite{hotstuff}, to
reproduce the asymptotically optimal quadratic asynchronous byzantine
agreement solution of VABA~\cite{vaba}.


\paragraph{Roadmap.}
The rest of the paper is organized as follows:
Section~\ref{sec:model-and-problem} describes the model and
formalizes the agreement and SMR problems. 
Section~\ref{sec:VbV} gives an overview of the view-by-view paradigm,
capturing its core properties and vulnerabilities. 
Section~\ref{sec:framework} defines ACE's abstractions, and its
algorithms are given in Section~\ref{sec:algorithms}. 
Section~\ref{sec:evaluation} instantiates ACE and evaluates its
performance.
Finally, Section~\ref{sec:related-work} discusses related work and
Section~\ref{sec:discussion} concludes.

\section{Model and Problem Definitions}
\label{sec:model-and-problem}

Section~\ref{sub:model} formally define the model and the
Agreement and SMR problems are defined in
Section~\ref{sub:problem-definition}.

\subsection{System Model}
\label{sub:model}

\paragraph{Communication.}

We consider a peer to peer system with $n$ parties, $f<n$ of which may
fail.
We say that a party is \emph{faulty} if it fails at any time during an
execution of a protocol.
Otherwise, we say it is \emph{correct}.
In a peer to peer system every pair of parties is connected with a
\emph{communication link}.
A message sent on a link between two correct parties is guaranteed to
be delivered, whereas a message to or from a faulty party might be
lost.
A link between two correct parties is \emph{asynchronous} if the
delivery of a message may take arbitrary long time, whereas
a link between two correct parties is \emph{synchronous} if there is
a bound $\Delta$ for message deliveries.
In \emph{asynchronous network periods} all links among correct parties
are asynchronous, whereas during \emph{synchronous network periods}
all such links are synchronous.

A standard communication model assumed by algorithms that follow
the view-by-view paradigm is the \emph{partially synchronous}
model\footnote{Sometimes referred as eventual synchrony in the
literature.}~\cite{dwork1988consensus}.
In this model, there is an unknown point in every execution, called
\emph{global stabilization time (GST)}, which divides the execution
into two network periods: before GST the network is asynchronous
and after GST the network is synchronous.
The partially synchronous model was defined to capture spontaneous
network disconnections in wide-area networks, in which case it is
reasonable to assume that asynchronous periods are short and
synchronous periods are long enough for the protocols to make
progress.

However, the partially synchronous model fails to capture malicious
attacks that intentionally try to sabotage progress, and thus are
not suitable for many current use cases (e.g., Blockchains).
For example, one possible attack is the \emph{weakly adaptive
asynchronous} in which an attacker adaptively blocks one party at a
time from sending or receiving messages (e.g., via DDOS).
This results in a \emph{mobile} asynchrony that moves from
party to party, violating the GST assumption made by the partially
synchronous model, and thus prevents progress from all leader-based
view-by-view algorithms.

ACE, in contrast, assumes the fully asynchronous communication model,
and thus progress in network speed under all network conditions and
attacks as long as messages among correct parties are eventually
delivered.

%
%
%
%

\paragraph{Failures, cryptography, and cetera.}

As mentioned in the Introduction and explained in more detail below,
ACE abstracts away specific model assumptions and
implementation details into three primitives: \emph{Leader based
view (LBV)}, \emph{leader-election}, and \emph{barrier}.
In Section~\ref{sec:framework}, we define the properties of these
primitives and require that any leader-based view-by-view protocol
that is instantiated into our framework satisfies them.
To satisfy these properties, each protocol may have different
model assumptions:
for example, the relation between $f$ and $n$, the type of
failures that may occur (e.g., crash and byzantine), and
cryptographic assumptions.
ACE inherits the specific assumptions made by
each of the protocols it is instantiated with, and
adds nothing to them. In other words, whatever assumptions are made by the instantiated
protocol in order to satisfy the abstractions' properties, are
exactly the assumptions under which ACE operates.



%
%


\subsection{Problem Definition}
\label{sub:problem-definition}
We now formally define the problems ACE implements.
We start with a fair validated single-shot agreement definition, and
then define the SMR problem which is a generalization of the single
shot agreement into a multi-shot problem.


\paragraph{Fair validated single-shot agreement.}
The \emph{fair validated agreement}~\cite{vaba, Cachin2000RandomOI,
CachinSecure} is single-shot problem in which correct parties propose
externally valid values and agree on one unique such value.
The formal properties are given below:  

\begin{itemize}
    
  \item Agreement: All correct parties that decide, decide on the
  same value.

  \item Termination: If all correct parties propose 
  valid values, then all correct parties decide with probability $1$.
  
  \item Validity: If a correct party decides an a value $v$, then $v$
  is externally valid.

\end{itemize}

\noindent Note that the agreement and termination properties are not
enough by them self to guarantee real progress of any multi-shot
agreement system (e.g., Blockchain) that is built on top of the
single-shot problem.
Without external validity, parties are allowed to agree on some
pre-defined value (i.e., $\bot$)~\cite{mostefaoui2017signature}, which is
basically an agreement not to agree.
Moreover, as long as a value satisfies the system's external validity
condition (e.g., no contradicting transactions in a blockchain
system), parties may decide on this value even if it was proposed by a
byzantine party.
However, since high stake is involved and byzantine parties may
try to increase the ratio of decision values proposed by them, we
require an additional fairness property that is a generalization of
the \emph{quality} property defined in~\cite{vaba}:

\begin{itemize}
    
  \item Fairness: The probability for a correct party to decide on a
  value proposed by a correct party is at least $1/2$. Moreover,
  during synchronous periods, all correct parties have an equal
  probability of $1/n$ for their values to be chosen. 
  
\end{itemize}

\noindent Intuitively, note that by simply following the protocol
byzantine parties can have a probability of $1/3$ (recall that $1/3$
of the parties are byzantine) for their value to be chosen in every
protocols even during synchronous periods. 
And since during asynchronous periods the adversary can, in addition,
block $1/3$ of the correct parties, we get that byzantine parties can
increase their probability to $1/2$.
Meaning that the fairness property we require is optimal.

\paragraph{Fair state machine replication.}

A state machine replication~\footnote{Sometime referred to as atomic
broadcast~\cite{CachinOPODIS}} (SMR) is a generalization of a
single-shot agreement problem into a multi-shot agreement
system~\cite{lamport2001paxos}.
Informally, an SMR system agrees on a (possibly infinite) sequence of
valid values. 
Formally, every correct party inputs with a (possible infinite)
sequence of externally valid values, and outputs a sequence of tuples
of the form $\lr{sq,v}$, where $sq \in \{1,…,N\}$ for an arbitrary
large $N$ and $v$ is an (externally valid) value.
Note that each output event contains exactly one such tuple, and in
in the rest of the paper we sometimes refer to the sequence number
$sq$ as a \emph{slot} number.
%
%
An implementation of a fair \emph{SMR} must satisfies the following
properties:

\begin{itemize}
  
  \item Integrity: For every $sq \in \{1,\ldots,N\}$, a correct party
  outputs at most one tuple $\lr{sq,v}$.
  
  \item Validity:  For every $sq \in \{1,\ldots,N\}$, if a correct
  party outputs a tuple $\lr{sq,v}$, then $v$ is externally valid.
  
  \item Termination: For every $sq \in \{1,\ldots,N\}$, all correct
  parties eventually output a tuple $\lr{sq,v}$ for some $v$ with
  probability $1$.
  
  \item Agreement: For every $sq \in \{1,\ldots,N\}$, if two correct
  parties output $\lr{sq,v}$ and $\lr{sq,v'}$, then $v = v'$.
  
  \item Fairness:  For every $sq \in \{1,\ldots,N\}$, the probability
  for a correct party to output $\lr{sq,v}$ s.t. $v$ was proposed by
  a correct party is at least $1/2$.
  
\end{itemize}

\noindent In order to capture the requirements made by systems
like Blockchains in which the validity of a value proposed in slot $i$
depends on all the decision values from slots 1 to $i-1$, we add an
additional property to our SMR definition: 

\begin{itemize}
 
  \item FIFO: For every $sq \in \{2,\ldots,N\}$, if a correct party
  $p$ outputs a tuple $\lr{sq,v}$ for some $v$, then for every $1
  \leq sq' < sq$ $p$ previously outputted $\lr{sq',v'}$ for some~$v'$.
\end{itemize}

\noindent Moreover, while from a theoretical point of view we usually
only care about the total resources consumed by correct parties up to
the point when they all decide, ignoring the resources used from this point on.
From a practical point of view, systems
wish to garbage collect all resources allocated for a specific slot
immediately after a decision for this slot has been made.
Since it is not always straight forward to map resources into slots in
which they are consumed, we capture the above by requiring following
property:

\begin{itemize}
 
  \item Strong halting: For any arbitrary large $N$,
  eventually all resources are de-allocated.  
\end{itemize}

%
%
%

%
%

\noindent Note that the standard way to achieve halting
is by the so called \emph{state
transfer} in which parties reliably broadcasts the decision value of
each slot
to all parties before
de-allocating all the associated resources and moving to the next slot.
This approach requires quadratic communication and is the
approach we take in this paper.
Dolev and Strong~\cite{dolev1983authenticated} have shown that even
in a synchronous setting this quadratic communication is unavoidable,
so at least asymptotically, this does not introduce a new overhead.

\section{The View-by-View Paradigm}
\label{sec:VbV}

Many (if not all) practical agreement and consensus
algorithms operate a leader-based view-by-view paradigm, designed for
partially synchronous models, including the seminal work of Dwork et
al.~\cite{dwork1988consensus} pioneering the approach, and underlying
classical algorithms like Paxos~\cite{lamport2001paxos},
Viewstamped-Replication~\cite{oki1988viewstamped}, PBFT~\cite{pbft},
and others~\cite{zyzzyva,ongaro2014search}.

Protocols designed according to the view-by-view paradigm advance in
views.
Every view has a designated leader that proposes a value and tries to
convince other parties to decide on it.
In order to tolerate faulty leaders from halting progress forever,
parties use timers to measure leader progress; if no
progress is made they demote the leader, abandoning
the current view and proceeding to the next one.


The main problem with this approach is that a faulty leader that does
not send any messages is indistinguishable from a correct leader with
asynchronous links.
Therefore, protocols implementing this approach are not able to
guarantee progress during asynchronous periods or weakly adaptive
asynchronous attacks since parties advance views before correct
leaders are able to drive decisions.

\subsection{Core properties}
\label{sub:core}

\paragraph{Safety.}

Perhaps the most important property of algorithms designed according
to the view-by-view paradigm is their ability to satisfy safety
during arbitrary long asynchronous periods.
This is achieved via a careful \emph{view-change} mechanism that
governs the transition between views.
View-change consists of parties \emph{wedging} the current
view by abandoning the current leader, and \emph{exchanging}
information about what might have committed in the view (the closing state of
the view).
In the new view, parties participate in the new leader's
phase only if it proposes a value that is safe in accordance with the
closing state.

\paragraph{Liveness.}

Algorithms that rely on leaders to drive progress cannot
guarantee progress during asynchronous periods since they cannot
distinguish between faulty leaders and correct ones with
asynchronous links.
During asynchronous periods, messages from
the current leader may be delivered only after parties timeout and
move to the next view regardless of how conservative the timeouts
are set.

However, all these algorithms share an important property that
our framework utilizes: for every view, if the leader of the view is
correct and no correct party times out and abandons this view, then
all correct parties decide in this view.


\subsection{Practical Vulnerabilities}
\label{sub:vulnerabilities}

Deploying view-by-view algorithms requires tuning the leader
timeouts.
On the one hand, aggressive timeouts set close to the common network delay might 
cause correct leaders to be demoted due to spurious delays, and destabilize the
system. On the other, conservative timeouts implies delayed actions in case of
faulty leaders. It further opens the system to possible attacks by
byzantine leaders that slow system progress to the maximum possible
without triggering a timeout.

Another attack on the progress of leader-based protocols is the weak
adaptive asynchrony in which an attacker blocks
communication with the leader of each view until the view expires,
e.g., via distributed denial-of-service attack.
Last, a carefully executed adaptive asynchrony attack can cause a
fairness bias. 
Some leaders (possibly byzantine) may be allowed to
progress and commit their values, whereas an attacker blocks
communication with other designated (possibly all correct) leaders.
In Section~\ref{sec:evaluation}, we demonstrate the above attacks,
and show that ACE is resilient
against them.


\section{Framework abstractions}
\label{sec:framework}

ACE provides
``asynchronous boosting'' for partially
synchronous protocols designed 
according to the leader-based view-by-view paradigm.
In a nutshell, ACE takes such a protocol, encapsulates a single view
of the protocol into a \emph{leader-based view (LBV)} abstraction
that provides API to avoid timeouts, composes LBVs into a wave of $n$
instances running in parallel, interjects auxiliary actions in
between successive waves, and chooses one LBV instance
retrospectively at random.
Detailed description is given in the next section.
Section~\ref{sub:LBV} defines the
\emph{Leader based view (LBV)} abstraction 
and Section~\ref{sub:auxiliary} introduces auxiliary abstractions
utilized by ACE.

\subsection{Encapsulating view-based agreement protocols}
\label{sub:LBV}


As explained above, each view in a leader-based view-by-view
algorithm consists of two phases: First, all parties wait for the
leader to perform the \emph{leader-based} phase to drive
decision on some value $v$, and then, if the leader fails to do it fast enough,
parties switch to the \emph{view-change} phase in which they \emph{wedge}
the current leader and \emph{exchange} information in order to get the
closing state of the view.
To decide when to switch between the phases, existing algorithms use
timeouts, which prevent them from guaranteeing progress during
asynchronous periods.
Therefore, in order to boost asynchronous liveness, ACE replaces the
timeout mechanism with a different strategy to switch between the phases.  
To this end, the LBV abstraction exposes an API with two
methods, \propose and \vc, where \propose starts the first
phase of the view (leader-based), and \vc switches to the second
(view-change).
By exposing API with these two methods, we remove the
responsibility of deciding when to switch between the phases from the view
(e.g., no more timeouts inside a view) and give it to the framework,   
while still preserving all safety guarantees provided by each view in
a leader-based view-by-view protocol.


Every instance of the LBV abstraction is parametrized with the
leader's name and with an identification $id$, which contains
information used by the high-level agreement algorithm built (by the
framework) on top of a composition of LBP instances. 
The \vc method gets no parameters and returns a tuple $\lr{s,v}$,
where $v$ is either a value or $\bot$; and $s$ is the closing state of
the instance, which consists of all the necessary information
required by the specific implementation of the abstraction (e.g., a
safe value for a leader to propose and a validation function all parties use to check
if the proposed value is safe). 
The \propose method gets 
the ``closing state'' $s$ that was returned from \vc in the
preceding LBV instance (or the initial state in case this is the
first one), and outputs a value $v$. 
Intuitively, the returned value from both methods is the ``decision''
that was made in the LBV instance, but as we explain below, the
high-level agreement algorithm might choose to ignore this value.

The safety of view-by-view algorithms strongly relies
on the fact that correct parties start a new view with the closing
state of the previous one.
Otherwise, they cannot guarantee that correct parties that decide in
different views decide on the same value. 
Therefore, when we encapsulate a single view in our LBV abstraction
and define its properties, we consider only executions in which the
LBV instances are composed one after another. 
Formally, we say that the LBV abstractions are \emph{properly
composed by a party $p_i$} in an execution if $p_i$ invokes the
\propose of the first instance with some fixed initial state (which
depends on the instantiated protocol), and for every instance $k>1$,
$p_i$ invokes its \propose with the state output of \vc of instance
$k-1$.
In addition, we say that the LBV abstractions are \emph{properly
composed} in an execution if they are properly composed by all correct
parties.
Figure~\ref{fig:LBVproperUsage} illustrates LBV's API and its properly
composed execution.

\begin{figure}[th]
    \centering
        \includegraphics[width=0.48\textwidth]{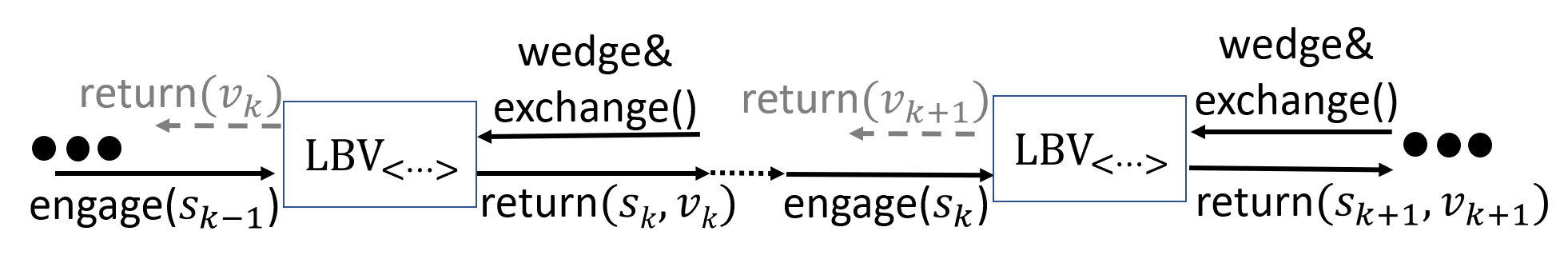}
        \caption{A properly composed execution: The \propose method of
        instance $k>1$ gets the state output of the \vc method of
        instance $k-1$.}
    \label{fig:LBVproperUsage}
\end{figure}

The formal definition of the LBV abstraction is as follows:

\begin{definition}

A protocol implements an LBV abstraction  if the following properties
are satisfied in every properly composed execution that consists of a
sequence of LBV instances:
 
\paragraph{Liveness:}
 
\begin{itemize}
 
    \item \textbf{Engage-Termination:} For every instance with a
  correct leader, if all correct parties invoke \propose and no
  correct party invokes \vc, then \propose invocations by all correct
  parties eventually return.
  
  \item \textbf{Wedge\&Exchange-Termination:} For every instance, if
  all correct parties invoke \vc then all \vc by correct parties
  eventually return.

\end{itemize}


\paragraph{Safety:}

\begin{itemize}
  
  \item \textbf{Validity:} For every instance, if an \propose or\\ \vc
  invocation by a correct party returns a value $v$,
  then $v$ is externally valid.
  
  \item \textbf{Completeness:} For every instance, if $f+1$
  \propose invocations by correct parties return, then
  no \vc invocation by a correct party returns a value $v =
  \bot$.
  
  \item \textbf{Agreement:} If an \propose or \vc  invoked in some
  instance by a correct party returns a value $v \neq \bot$ and some
  other \propose or \vc invoked in some instance by a correct party
  returns $v' \neq \bot$ then $v = v'$.

\end{itemize}

\end{definition}


Note that during the view-change phase in most leader-based
protocols, parties send the closing state only to the leader of the
next view.
However, in ACE, since we run $n$ concurrent LBV instances,
each with a different leader, we need all parties to learn the
closing state after \vc returns.
Moreover, as mentioned above and captured by the Completeness
property, we use \vc to also boost decisions in order to guarantee
that if the retrospectively chosen LBV instance successfully
completed the first (leader-based) phase, than all correct parties
decide at the end of its second phase.
Therefore, when encapsulating the view-change mechanism of a
leader-based protocol into the \vc method, a small change has to
be made in order to satisfy the above properties.
Instead of sending the closing state only to the next leader, parties
need to exchange information by sending the closing state to all parties
and wait to receive $n-f$ such messages.
No change is needed to the first phase of the
encapsulated leader-based protocol since
all the required properties for \propose are implicitly satisfied.


\subsection{Auxiliary abstractions}
\label{sub:auxiliary}

We now define two additional abstractions required by ACE.
Similarly to the LBV abstraction, each instance is parametrized with
an identification $id$, and the implementation details and 
model assumptions are abstracted away.


\paragraph{Barrier.}
The barrier abstraction is used to synchronize
between parties, ensuring that correct parties wait for each other before progressing.
The abstraction exposes two API methods: \barrierready and
\barriersync.
Below we define the properties each barrier implementation must
satisfy:

\begin{definition}

An implementation of the Barrier abstraction must satisfy the
following properties:

\begin{itemize}
  
  \item \textbf{B-Coordination:} No \barriersync invocation by a
  correct party returns before $f+1$ correct parties invoke
  \barrierready.
  
  \item \textbf{B-Termination:} If all correct parties invoke
  \barrierready then all \barriersync invocations by correct
  parties eventually return.
  
  \item \textbf{B-Agreement:} If some \barriersync invocation by
  a correct party returns, then all \barriersync invocations by
  correct parties eventually return.
  
\end{itemize}

\end{definition}

\paragraph{Leader-election.}
Our Leader-election abstraction is similar to the one defined
in~\cite{vaba}, which exposes one operation, \elect, to elect
a unique leader.
The formal properties are given below.

\begin{definition}

An implementation of the Leader-election abstraction must satisfy the
following properties:

\begin{itemize}
  
  \item \textbf{L-Termination:} If $f+1$ correct parties invoke
  \elect, then all \elect invocations by correct parties return.
  
  \item \textbf{L-Agreement:} All invocations of \elect by correct
  parties that return, return the same party.
  
  \item \textbf{L-Validity:} If an invocation of \elect by a
  correct party returns, it returns a party $p_i$ with probability $1
  / n$ for every $p_i \in \{p_1,\ldots,p_n\}$.
  
  \item \textbf{L-Unpredictability:} The probability of the adversary
  to predict the returned value of an \elect invocation by a
  correct party before any correct party invokes \elect is at most
  $1 / n$.
  
\end{itemize}

\end{definition}

%
%
%
%
%
%
%
%

\noindent In Section~\ref{sec:evaluation} we give example
implementations of these abstractions in the byzantine failure model.

\section{Framework algorithms}
\label{sec:algorithms}

In this section we present ACE's asynchronous boosting
algorithms, which are built on top of the abstractions defined above.
We first present in Section~\ref{sub:singleShot} an algorithm for an
asynchronous single-shot agreement, and then, in
Section~\ref{sub:SMR} we show how to turn it into an asynchronous SMR.
For completeness, in Section~\ref{sec:recAlg}, we show how to
use the LBV abstraction to reconstruct a variant of the base
partially synchronous view-by-view algorithm the LBV is instantiated
with.

\subsection{Asynchronous fair single-shot agreement}
\label{sub:singleShot}

The pseudocode for the asynchronous single-shot agreement
protocol appears in Algorithm~\ref{alg:AS} and a formal
correctness proof is given in Section~\ref{sec:proof}.
An invocation of the protocol (\emph{SS-propose}$(id,S)$) gets an
initial state $S$ and identification $id$, where the initial state
$S$ contains all the initial specific information (including the
proposed value) required by the leader-based view-by-view protocol
instantiated in the LBV abstraction.

\begin{algorithm}
\caption{Asynchronous single-shot agreement.} 
\begin{algorithmic}[1]
\footnotesize

\Upon{\emph{SS-propose(id,S)}}

\State $state \gets S$ ; $wave \gets 1$

\While {true}

	\State $ID \gets \lr{id,wave}$
	\State $\lr{state',value} \gets \Call{wave}{ID, state}$

	\If{$value \neq \bot$ and did not decide before}
	
		\State \textbf{decide} $\lr{id,value}$
	
	\EndIf
	\State $state \gets state'$
	\State $wave \gets wave+1$ 
\EndWhile

\EndUpon

\Statex

\Procedure{wave}{$ID, state$}
  
	
	\ForAll{$p_j=p_1,\ldots,p_n$}
		\State invoke \propose$_{\lr{ID,p_j}}(state)$
		\Comment{non-blocking}
	\EndFor
	
	\State \barriersync$_{ID}()$
	
	
	\State $leader  \gets $ \emph{$elect_{ID}()$}


	\State \textbf{return} \vc$_{\lr{ID,leader}}()$
	

\EndProcedure

\Statex 

\Upon{\propose$_{\lr{ID,p_j}}$ returns $v$}
 
 	\State send ``ID, $\textsc{engage-done}$'' to party $p_j$
 
\EndUpon

\Statex

\Upon{receiving $n-f$ ``ID,$\textsc{engage-done}$'' messages }
 
	\State invoke \barrierready$_{ID}()$
 
\EndUpon

\end{algorithmic}
\label{alg:AS}
\end{algorithm}

The protocol proceeds in a \emph{wave-by-wave} manner. 
The state is updated
at the and of every wave and a decision is made the first time a wave
returns a non-empty value.
In every wave, each party first invokes the \propose
operation in $n$ LBV instances, each with a
different leader.
Each invocation gets the state obtained at the end of the previous
wave or the initial state if this is the first wave.

Then, parties invoke \barriersync and wait for it to return.
Recall that by the B-Coordination property, \barriersync
returns only after $f+1$ correct parties invoke \barrierready.
When an \propose invocation in an LBV instance with leader $p_j$
returns, a correct party sends an ``$\textsc{engage-done}$''
message to party $p_j$, and whenever a party gets $n-f$ such messages
it invokes \barrierready.
Denote an LBV instance as \emph{successfully completed} when $f+1$
correct parties completed the first phase, i.e., their \propose
returned, and note, therefore, that a correct party invokes
\barrierready only after the LBV instance in which it acts as
the leader was successfully completed.
Thus, a \barriersync invocation by a correct party returns
only after $f+1$ LBV instances successfully completed.

\begin{figure}[th]
    \begin{center}
        \includegraphics[width=0.48\textwidth]{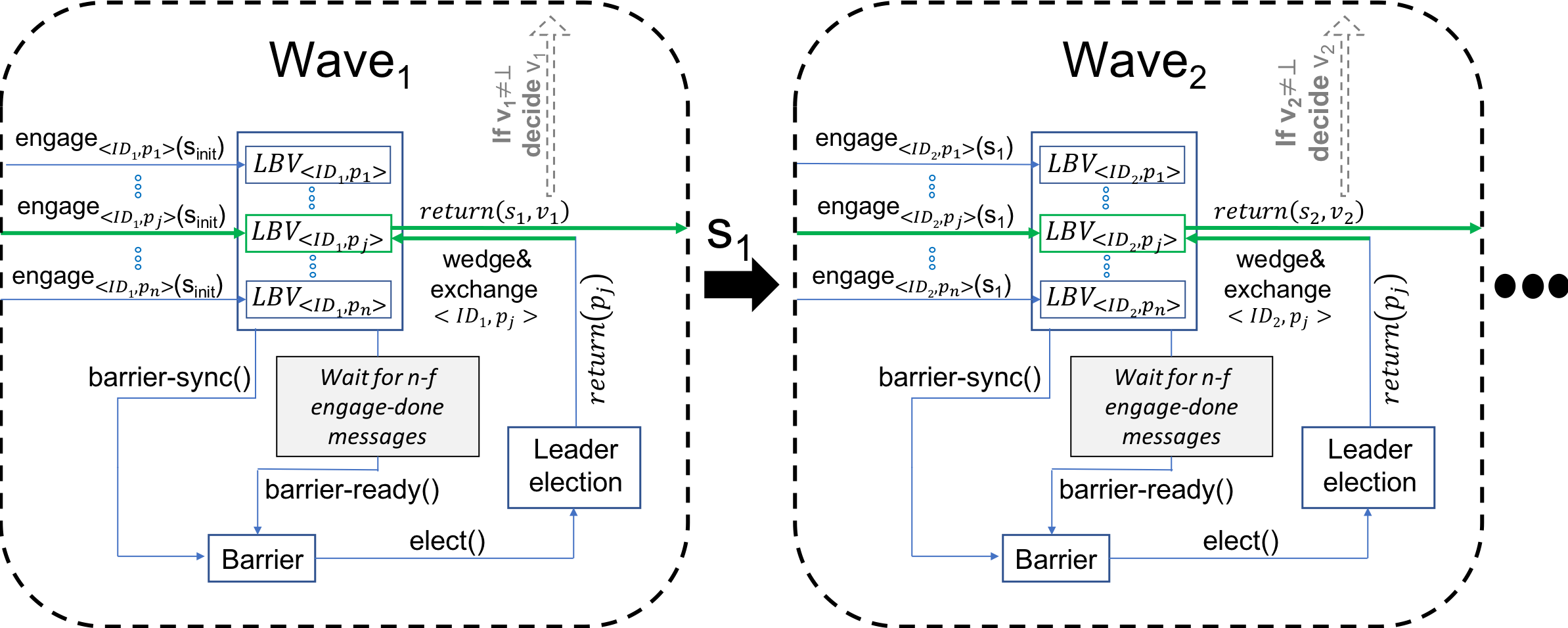}
    \end{center}
    \caption{Asynchronous single-shot
    algorithm. The chosen LBVs, which are marked in
    green, are properly composed.}
    \label{fig:ASprotocol}
\end{figure}  

Next, when the \barriersync returns, parties elect a unique
leader via the leader-election abstraction, and further consider
only its LBV instance.
Note that since parties wait until $f+1$
LBV instances have successfully completed before electing the
leader, with a constant probability of $\frac{f+1}{n}$ the
parties elect a successfully completed instance\footnote{This can
be improved to $\frac{2f+1}{n}$ in the byzantine case with $n=3f+1$
if parties attach completeness proofs to \barrierready messages.}, and
even an adaptive adversary has no power to prevent it.

Finally, all parties invoke \vc in the elected LBV instance to wedge
and find out what happened in its first phase, using the returned
state for the next wave and possibly receiving a decision value.
By the Completeness property of LBV, if a successfully completed LBV
instance is elected, then all \vc invocations by correct parties
return $v \neq \bot$ and thus all correct parties decide $v$ in this
wave.
Therefore, after a small number of $\frac{n}{f+1}$ waves all
correct parties decide in expectation.
Note that the sequence of chosen LBV instances form a properly
composed execution, and thus since parties return only values
returned from chosen LBVs, our algorithm inherits its safety
guarantees from the leader-based protocol the LBV is instantiated
with.
An illustration of the algorithm appears in
Figure~\ref{fig:ASprotocol}.

\subsubsection{Correctness Proof}
\label{sec:proof}

We prove that Algorithm~\ref{alg:AS} satisfies 
Agreement, Termination, Validity, and Fairness properties
in the asynchronous communication model.
We start by proving the Agreement and Validity properties.

\begin{observation}
\label{obs:compose}

The chosen LBV instances in Algorithm~\ref{alg:AS} form a properly
composed execution.

\end{observation}

\begin{lemma}
\label{claim:AS:V_A}

Algorithm~\ref{alg:AS} satisfies \emph{Validity} and
\emph{Agreement}.

\end{lemma}

\begin{proof}

By the code, correct parties only decide on values returned from a
\vc method invoked on one of the chosen instances.
Therefore, by Observation~\ref{obs:compose}, the Validity and Agreement
properties follow from the Safety properties required by the LBV
abstraction.

\end{proof}

\noindent We now prove that Algorithm~\ref{alg:AS} satisfies
Termination.
We start by showing that no honest party is stuck forever in a
wave.

\begin{claim}
\label{claim:AS:ready}

Consider a wave $k \geq 1$ that all correct parties start, then at
least one \barriersync invocation by a correct party eventually
returns.

\end{claim}

\begin{proof}

Assume by way of contradiction that no \barriersync()
invocations by correct party eventually returns in wave $k$.
Thus, no correct party ever invoke \vc in wave $k$.
Therefore, since all correct parties invoke \propose in all
LBV instances in wave $k$,
we get by the Engage-Termination property that all
\emph{propose} invocation by correct parties eventually return, and
thus all correct parties eventually send $\textsc{engage-done}$
messages to all correct leaders.
Hence, all correct leaders eventually get $n-f$ $\textsc{engage-done}$
messages, and thus eventually invoke \barrierready.
The contradiction follows from the B-Termination property.

\end{proof}

\begin{claim}
\label{claim:AS:view}

For every wave $k \geq 1$, if all correct parties start wave $k$, then
all correct parties eventually complete wave $k$.

\end{claim}

\begin{proof}

By Claim~\ref{claim:AS:ready}, some \barriersync
invocation by a correct party eventually returns in wave $k$.
Thus, by the B-Agreement property, all \barriersync
invocations by correct parties eventually return in wave $k$, and
thus all correct parties eventually invoke \elect in wave $k$.
By the L-Termination property, all \elect invocations by
correct parties eventually return in wave $k$.
Therefore, the Claim follows from the Wedge\&Exchange-Termination
property.

\end{proof}

\noindent The next corollary follows by inductively applying
Claim~\ref{claim:AS:view}.

\begin{corollary}
\label{col:viewcomplete}

For every $k \geq 1$, all parties eventually complete wave $k$ in
Algorithm~\ref{alg:AS}.

\end{corollary}

\noindent For the rest of the proof we say that an LBV instance is
\emph{completed} if at least $f+1$ \propose innovations by correct
parties previously returned.
We next bound the probability to choose a leader of a completed LBV
instance.

\begin{claim}
\label{claim:AS:leader}

For every wave $k \geq 1$, at least $f+1$ LBV instances with correct
leaders are completed before some correct party invokes \elect.

\end{claim}

\begin{proof}

Consider some correct party that invokes \elect at wave
$k$.
By the code, its \barriersync invocation was previously
returned, and thus by the B-Coordination property, at least $f+1$
correct parties previously invoked \barrierready.
Therefore, at least $f+1$ correct parties
received $n-f$ $\textsc{engage-done}$ messages, at least
$f+1$ of which are from correct parties.
Since correct parties send $\textsc{engage-done}$ messages to party
$p_j$ only after their \propose invocation in the LBV instance in
which $p_j$ acts as leader returns, at least $f+1$ LBV
instances with correct leaders are completed before some correct party
invokes \elect.

\end{proof}

\begin{claim}
\label{AS:probability}

Consider a wave $k \geq 1$ that all parties start, the probability for
all correct parties to decide at wave $k$ is at least $\frac{f+1}{n}$.

\end{claim}

\begin{proof}

By Corollary~\ref{col:viewcomplete}, all correct parties invoke
\elect in view $k$ and, by the L-Agreement property, all
invocations return the same leader.
By Claim~\ref{claim:AS:leader}, at least $f+1$ LBV instances are
completed before some correct party invokes \elect at view $k$.
Therefore, by the L-Validity and L-Unpredictability properties, we get
that the probability to choose a leader of a completed LBV instance is
at least is $\frac{f+1}{n}$.
By the Completeness property of the LBV abstraction, if a completed
LBV instance is chosen, then no \vc invocation by correct party
returns $v= \bot$.
Therefore, by the code, all correct parties decide at the end of wave
$k$ with probability $\frac{f+1}{n}$.


\end{proof}

\begin{lemma}

Algorithm~\ref{alg:AS} satisfies Termination. 

\end{lemma}

\begin{proof}

By Corollary~\ref{col:viewcomplete}, correct parties are never
stuck indefinitely in any view, and thus all honest parties eventually start
all views.
Therefore, by Claim~\ref{AS:probability}, all honest parties
eventually decide with probability 1.
Moreover, the expected number of views after which all correct party decide
is $\frac{n}{f+1}$.

\end{proof}

\begin{lemma}

Algorithm~\ref{alg:AS} satisfies Fairness. 

\end{lemma}

\begin{proof}

By the Completeness property of the LBV abstraction, if a completed
instance is chosen, than all correct parties decide at the end of
the wave. 
By Claim~\ref{claim:AS:leader}, for every wave $k \geq 1$, at least
$f+1$ LBV instances with correct leaders are completed before a
unique instance is chosen in the wave.
Therefore, even if all LBV instances with faulty parties have been
completed as well, the probability to decide on value proposed by a
correct party is at least $1/2$. Moreover, during synchronous periods,
all instances with correct parties complete with equal probability,
and thus by the L-validity, all correct parties have an equal
probability for their values to be chosen.

\end{proof}

\subsection{Asynchronous fair state machine replication}
\label{sub:SMR}

The pseudocode for the asynchronous fair state machine replication
appears in Algorithm~\ref{alg:AS-SMR}.
The parameter $S$, passed to the \emph{SMR-propose} invocation, is a
vector consisting of an initial state for each slot. We use an instance of
the asynchronous single-shot agreement to agree on the value of each slot, 
and thus the SMR algorithm inherits the Integrity, Fairness, Validity
and Agreement properties from the single-shot one.
In order to satisfy FIFO, parties do not advance to the next
slot until they learn the decision of the current one, and in order to
satisfy Halting they de-allocate all resources associated with the
current slot and abandon the slot's single-shot algorithm once
they move.
Note that abandoning a single-shot algorithm might violate
its Termination because it relies on the participation of all
correct parties.
Therefore, to satisfy the Termination of the SMR, 
we use a forwarding mechanism to reliably broadcast the decision
value before abandoning the current slot and moving to the next one.

Since ACE is model agnostic, parties do not explicitly use
threshold signatures and decision proofs in the
forwarding mechanism.
Instead, in every slot, correct parties waits until they either
decide on a value $v$ (via the slot's single-shot agreement) or
receive $f+1$ ``\textsc{decide}'' messages from parties claiming they
have decided on a value $v$.
In the second case, the receiving party knows that at least one
correct party decided $v$ even in the byzantine failure model.
Then, to make sure that all correct parties eventually finish waiting
even though some might have moved to the next slot already, parties use
the barrier abstraction that provides the required guarantee with its
B-Coordination property.
Note that although the above forwarding mechanism works for both
byzantine and crash failures models, in the latter case the forwarding
mechanism can be simplified:
a party only needs to echo its decision value before it moves to
the next slot.
So in this case, the gray lines in the pseudocode can be dropped.  

\begin{algorithm}
\caption{Asynchronous fair state machine replication.} 
\begin{algorithmic}[1]
\footnotesize

\Upon{\emph{SMR-propose(S)}}

	\State $slot \gets 0$
	\For{\textbf{every} $s \in \mathbb{N}$}
	
		\CSTATE $M[s] \gets \{\}$
		\State $V[s] \gets \bot$; $D[s] \gets \emph{false}$

	\EndFor
	
	\While{true}
	
		\State $slot \gets slot+1$
		\State $ \emph{SS-propose}(slot,S[slot])$
		\State \textbf{wait} until $D[slot] = \emph{true}$ 
		\State send ``slot, $\textsc{decide}, value$'' to all parties
		\CSTATE \barrierready$_{slot}()$
		\CSTATE \barriersync$_{slot}()$
		\State \textbf{free} all resources associated with $slot$
		\State \textbf{output} $\lr{slot,V[slot]}$

	\EndWhile

\EndUpon

\Statex

\Upon{decide $\lr{slot,v}$}

	\State $V[slot] \gets v$; $D[slot] \gets \emph{true}$

\EndUpon

\Statex

\Receiving{``slot, $\textsc{decide}, v$'' from party $p_j$}

	\CSTATE $M[slot] \gets M[slot]\cup \{\lr{p_j,v}\}$
	\CIF{$\exists v$ s.t.\ $|\{p | \lr{p,v} \in M[slot] \}| = f+1$}
	
		\State $V[slot] \gets v$; $D[slot] \gets \emph{true}$ 
	
	\EndIf

\EndReceiving

\end{algorithmic}
\label{alg:AS-SMR}
\end{algorithm}

\subsection{Partially Synchronous view-by-view Agreement}
\label{sec:recAlg}

For completeness, we show how to use the LBV
abstraction to reconstruct a variant of the leader-based view-by-view
partially synchronous base Agreement algorithm that the LBV
abstraction is instantiated with.
To this end, we assume the base algorithm provides two methods:
\emph{getLeader} and \emph{getTimeout}.
These methods should implement the logic used by the base algorithm
to map designated leaders to views and set their timeouts,
respectively. 
In particular, \emph{getLeader(v)} gets a view $v$ and returns a
party $p_i \in \{ p_1\dots,p_n\}$, and \emph{getTimeout(v, S)} gets a view
$v$ together with a state $S$ and returns a timeout. 

The pseudocode appears in Algorithm~\ref{alg:reconSS}.
The protocol proceeds in views.
Timeouts are used in order to demote a leader who was unable to drive
progress.
At the beginning of every view, parties first get the leader and the
timeout of the current view, and then invoke \propose in the leaders
LBV instance and a timer to monitor the leaders progress.
If the \propose invocation returns $v$ before the timer expires, then
$v$ is decided.
In any case, whether the \propose invocation returns or
the timer expires, a \vc is invoked in order to update the state and
safely proceed to the next view.

Note that the algorithm forms a properly composed execution of the LBV
instances, and thus Validity and Agreement are trivially satisfied by
the Safety properties required by the LBV abstraction.
As any protocol in the partially synchronous model, the termination
of the algorithm requires a long enough synchronous period in which all
correct parties execute the same view. 


\begin{algorithm}
\caption{Reconstruction of base partially synchronous single-shot
agreement:
protocol for party $p_i$.} 
\begin{algorithmic}[1]
\footnotesize

\Upon{\emph{ES-propose(id,S)}}


\State $state \gets S$; $view \gets 1$
%
%
%

\While {true}
	\State $\emph{leader} \gets \Call{getLeader}{view}$
	\State $\emph{timeout} \gets \Call{getTimeout}{view, state}$
	\State invoke \propose$_{\lr{view,\emph{leader}}}(state)$
	\State invoke a timer to expire in $timeout$
 	\State \textbf{wait} for timer to expire or \propose to return
	\If{\propose returned $v$}
	
		\State \textbf{decide} $v$
	
	\EndIf

	\State $\lr{state,*} \gets $ \vc$_{\lr{view,\emph{leader}}}()$

	\State $view \gets view+1$ 
\EndWhile

\EndUpon

\Statex

%
%

\end{algorithmic}
\label{alg:reconSS}
\end{algorithm}

\section{ACE Instantiation}
\label{sec:evaluation}

There are many possible ways to instantiate the ACE framework.
We choose to evaluate ACE in the byzantine
failure model with $n=3f+1$ parties and a computationally bounded
adversary due to the attention it gets in
the Blockchain use-case. 
For the LBV abstraction, we implement a variant of
HotStuff~\cite{hotstuff}.
For the leader-election we implement the
protocol in~\cite{vaba, Cachin2000RandomOI}, and for the Barrier we
give an implementation that operates in the same model.
All protocols use a BLS threshold signatures
schema~\cite{boneh2001short} that requires a setup, which can be done
with the help of a trusted dealer or by using a protocol for an
asynchronous distributed key generation~\cite{kokorisbootstrapping}.
 


Our evaluation compares the performance of ACE's SMR
instantiated with HotStuff, we refer to as \emph{ACE HotStuff}, with
the base HotStuff SMR implementation.
To compare apples to apples, the base HotStuff and ACE HotStuff share
as much code as possible.
We present raw performance comparisons during synchronous periods,
and demonstrate the performance during asynchrony and under
adversarial attacks.

We proceed to describe our implementation of the ACE
building blocks in Section~\ref{sub:imp}, and in
Section~\ref{sub:ev}, we describe the environmental setup and
performance measurements.

\subsection{Implementation}
\label{sub:imp}

We implemented all algorithms in C++, and made use of a
BLS threshold signature~\cite{boneh2001short} implementation
provided in~\cite{concordOpenSource}.
Communication is done over TCP to provide reliable links.
We next describe our LBV, barrier, and leader election
implementations.
The communication complexity of a single LBV is linear and that of
the barrier and leader-election is quadratic,
leading to an expected total quadratic communication, for each
slot.


\paragraph{LBV.}
The instantiation of the LBV abstraction is the four-step view by
view algorithm of HotStuff~\cite{hotstuff}.
In the leader-based phase of each view in HotStuff, a leader drives a
decision in four steps of communication s.t.\ in every step the
leader sends a signed message (with the proposed value) to all
parties, which in turn verify, sign, and send it back to the leader.
To ensure linear communication, the leader utilizes threshold
signatures~\cite{boneh2001short} to provide concise proofs of a
quorum of signatures which are used for verification.
If parties timeout before the leader completes all four steps, they
start the view-change phase in which they send the closing state,
that consists of the messages they received in this view to the
leader of the next one.

To encapsulate the HotStuff protocol in the LBV abstraction we did
the following:
when a party invokes \propose, it begins participating in the
leader's four-step protocol as described above, and whenever it gets
the last step's message, it returns the value therein.
To verify the safety of the first step message, it
uses the information in the \emph{state} parameter passed to \propose. 
Note that if \vc is invoked before \propose returns, then \propose
never returns ($\mathsf{wedge}$\xspace).
In this case, parties behaves as if their timeouts expired in the
original HotStuff with the following change: instead of sending the
closing state only to the next leader, parties send it to all parties
($\mathsf{exchange}$\xspace)
since all parties act as leaders in the next wave.
When a party gets $n-f$ such closing states, it updates the 
\emph{state} parameter accordingly and outputs it.
In addition, if at least one received closing states contains a valid
four step's message, then the party also outputs the value therein.
Otherwise, it outputs $\bot$.
Note that we do not describe the specific logic of the four-step
protocol and the structure of the state -- an interested reader is
referred for more details to~\cite{hotstuff}.

Since the implementation keeps the safety logic of HotStuff, the
\emph{Validity} and \emph{Agreement} of LBV are 
satisfied. 
Since  \propose  implements the HotStuff leader's phase,
then if the leader is correct and no correct party invoke \vc, then
eventually all parties get the leaders four step's message and
return, thereby satisfying \emph{\propose-Termination}.
The \emph{\vc-Termination} is satisfied since it returns after
receiving $n-f$ closing states.
Af for \emph{Completeness}, if an \propose invocation returns a
value, than the invoking party gets the four step's message.
Therefore, if $f+1$ \propose invocation by correct parties return,
then all correct parties gets at least one closing state with the
four step's message during \vc, and return a value $v \neq
\bot$.

\paragraph{Barrier.}
When a party invokes \barrierready, it broadcasts a message with
its signature share for this barrier identification. 
When a party invokes \barriersync, it waits for the first
of the following two events.
If it receives $2f+1$ valid shares, it combines the shares into a
threshold signature, sends it to all other parties, and
returns.
If it receives a correct threshold signature, it forwards it to all
other parties and returns.

Since a \barriersync invocation does not return before receiving
$2f+1$ shares, then at least $f+1$ correct parties previously invoked
\barrierready, thus satisfying \emph{B-Coordination}. 
\emph{B-Termination} is satisfied since if all correct parties send
shares to all other parties, then every correct party
receives $2f+1$ shares, and can generate a valid threshold signature.
Due to the forwarding of threshold signatures, if one
correct party gets a threshold signature and returns, then all correct
parities eventually get it and return, satisfying
\emph{B-Agreement}.


\paragraph{Leader-election.}
The leader election is similar to that of \cite{vaba} and
\cite{Cachin2000RandomOI}.
When a party invokes \elect, it signs the instance's
identification and broadcasts its share.
When a correct party collects $f+1$ valid shares, it combines them
to a threshold signature, hashes the signature to get a pseudo
random value, and returns the value modulo $n$ to get a random leader.

As the threshold signature is generated from $f+1$ valid shares, all
correct parties eventually generate it provided that at least $f+1$
correct parties invoke \elect, satisfying \emph{L-Termination}.
Due to the nature of threshold cryptography, the generated
threshold signature is the same for all correct
parties~\cite{boneh2001short}, and since the hash and modulo functions
are deterministic, all parties agree on the electing leader,
satisfying \emph{L-Agreement}.
For formal analysis of \emph{L-Validity} and \emph{L-Unpredictability}
please refer to~\cite{Cachin2000RandomOI, cachin2001secure, vaba}.


\subsection{Evaluation}
\label{sub:ev}

We compare the SMR of ACE HotStuff with that of base
HotStuff. 
In Section~\ref{subsub:setup} we present the tests' setup.
Then, in Section~\ref{subsub:overhead}, we measure ACE's overhead
during failure-free synchronous periods, and in
Section~\ref{subsub:benefits} we demonstrate ACE's superiority
during asynchrony and attacks.

\subsubsection{Setup}
\label{subsub:setup}

We conducted our experiments using $\mathsf{c5d.4xlarge}$\xspace
instances on AWS EC2 machines in the same data center. 
We used between 1 and 16 virtual machines, each with 4 replicas.
The duration of every test was $60$ seconds, and every test was
repeated 10 times.
The size of the proposed values is $10000$ bytes. The latency is
measured starting from when a new slot has begun until a decision is made. The throughput is measured in one of two ways. In tests
where we altered the number of replicas, the throughout is the total number of
bytes committed, divided by the length of the test. In tests where we show the throughout
as a function of time, we aggregate the number of committed bytes in
$1$ second intervals.
We did not throttle the bandwidth in any run, rather we altered the transmission delays between
the machines, using $\mathsf{NetEm}$\xspace~\cite{netem}.

\subsubsection{ACE's overhead}
\label{subsub:overhead}

The first set of tests compare ACE HotStuff performance with that
of base HotStuff under optimistic, synchronous, faultless conditions. 
Figure \ref{fig:overhead} depicts the latency and throughput.
The delay on the links was measured to be under $1ms$.
The latency increases with the growth in the number
of replicas since
each replica must handle an equal growth in the number of
messages. Furthermore, as ACE HotStuff has a larger overhead than
base HotStuff, the latency grows faster.

Figure \ref{fig:overhead-delyed} shows the latency and throughput
with different delays added to the links. 
The latency of ACE HotStuff is twice that of base HotStuff.
This is expected, as ACE is expected
to execute 1.5 waves per slot, leading to 1.5x the latency.
Add on the additional barrier, leader election abstraction and we
arrive at 2x reduction in performance.

These tests show that the cost of using ACE is
about 2x reduction in performance in the optimistic case. 
In the next tests we argue this cost is sometimes worth paying,
as liveness of partially synchronous algorithms can be easily
affected.

\begin{figure}[th]
    \centering
    \begin{subfigure}[b]{0.23\textwidth}
        \includegraphics[width=\textwidth]{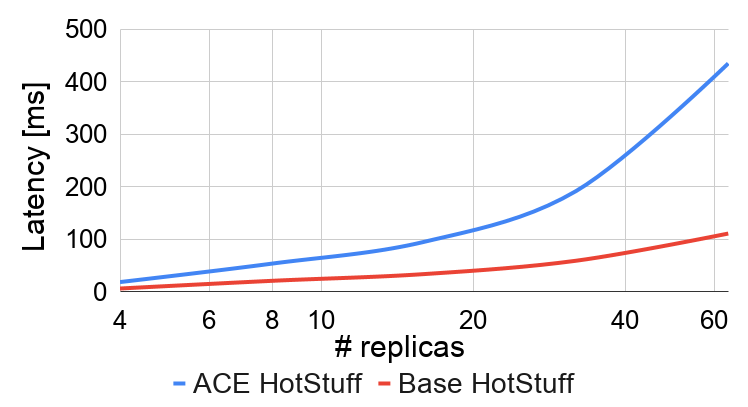}
          \caption{Latency}
        \label{fig:nodelaylat}
    \end{subfigure}
    \hfill
    \begin{subfigure}[b]{0.23\textwidth}
        \includegraphics[width=\textwidth]{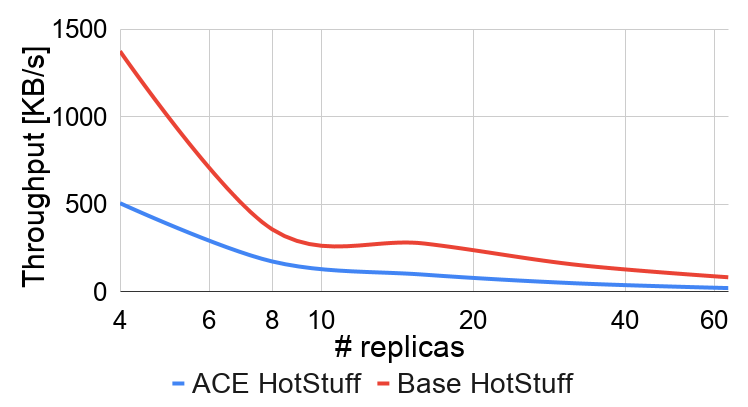}
          \caption{Throughput}
        \label{fig:nodelaythpt}
    \end{subfigure}
       \vspace*{-3mm}
        \caption{Optimistic case with no network delay.}
       \label{fig:overhead}
\end{figure}

\begin{figure}[th]
    \centering
    \begin{subfigure}[b]{0.23\textwidth}
        \centering
        \includegraphics[width=\textwidth]{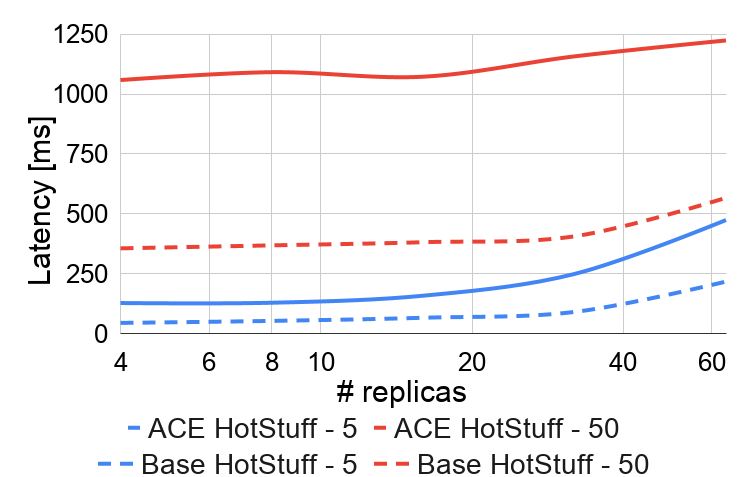}
          \caption{Latency}
        \label{fig:withdelaylat}
    \end{subfigure}
    \hfill
    \begin{subfigure}[b]{0.23\textwidth}
        \centering
        \includegraphics[width=\textwidth]{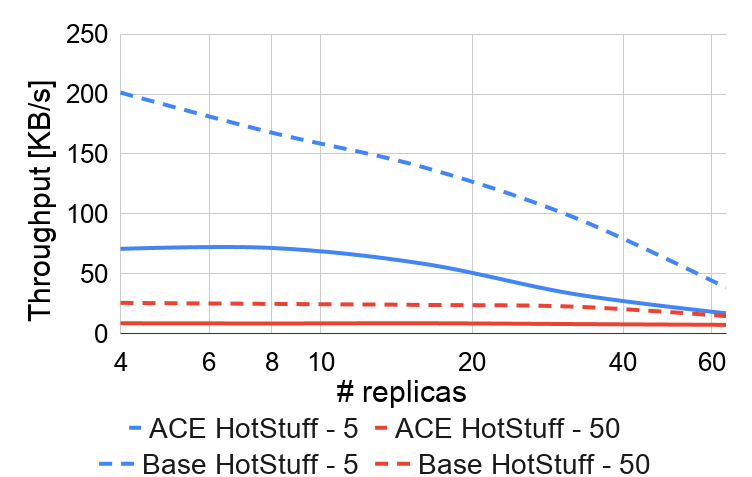}
          \caption{Throughput}
        \label{fig:withdelaythpt}
    \end{subfigure}
		 \vspace*{-3mm}
		\caption{Optimistic case under different network delays.}
       \label{fig:overhead-delyed}
\end{figure}

\subsubsection{ACE's superiority}
\label{subsub:benefits}

From here on we choose a configuration of 32 replicas and set the
transmission delay to be $5$ms unless specified otherwise.
The second set of tests compare ACE HotStuff and base HotStuff in
adverse conditions concerning message delays. 
These tests manipulate two factors, the transmission delays
(controlled via NetEm~\cite{netem}), and the view timeout
strategy. 

\paragraph{Periods of asynchrony.}

The first test sets base HotStuff view timers to a fixed
constant of $100$ms, the time needed for a commit assuming a $5$ms
transmission delay.
The test measures the performance drop during a short period in which
transmission delays are increased, simulating asynchrony.
%
%
For the first third of the test the network delay is $5$ms, for
the next third the delay is $10$ms, and finally the delay returns to
$5$ms. 

Figure \ref{fig:variabledelaythput} compares the throughput of ACE
HotStuff and base HotStuff. 
While the network delay is $5$ms, base HotStuff
outperforms ACE HotStuff. 
However, once the network delay begins to fluctuate, the throughput of
base HotStuff goes to $0$ since no leader has enough time to drive
progress.
ACE HotStuff only sees a drop in throughput proportional to
the delay, meaning that it continue to progress at network speed.

\sidecaptionvpos{figure}{c}
\begin{SCfigure}
      \includegraphics[width=0.28\textwidth,]{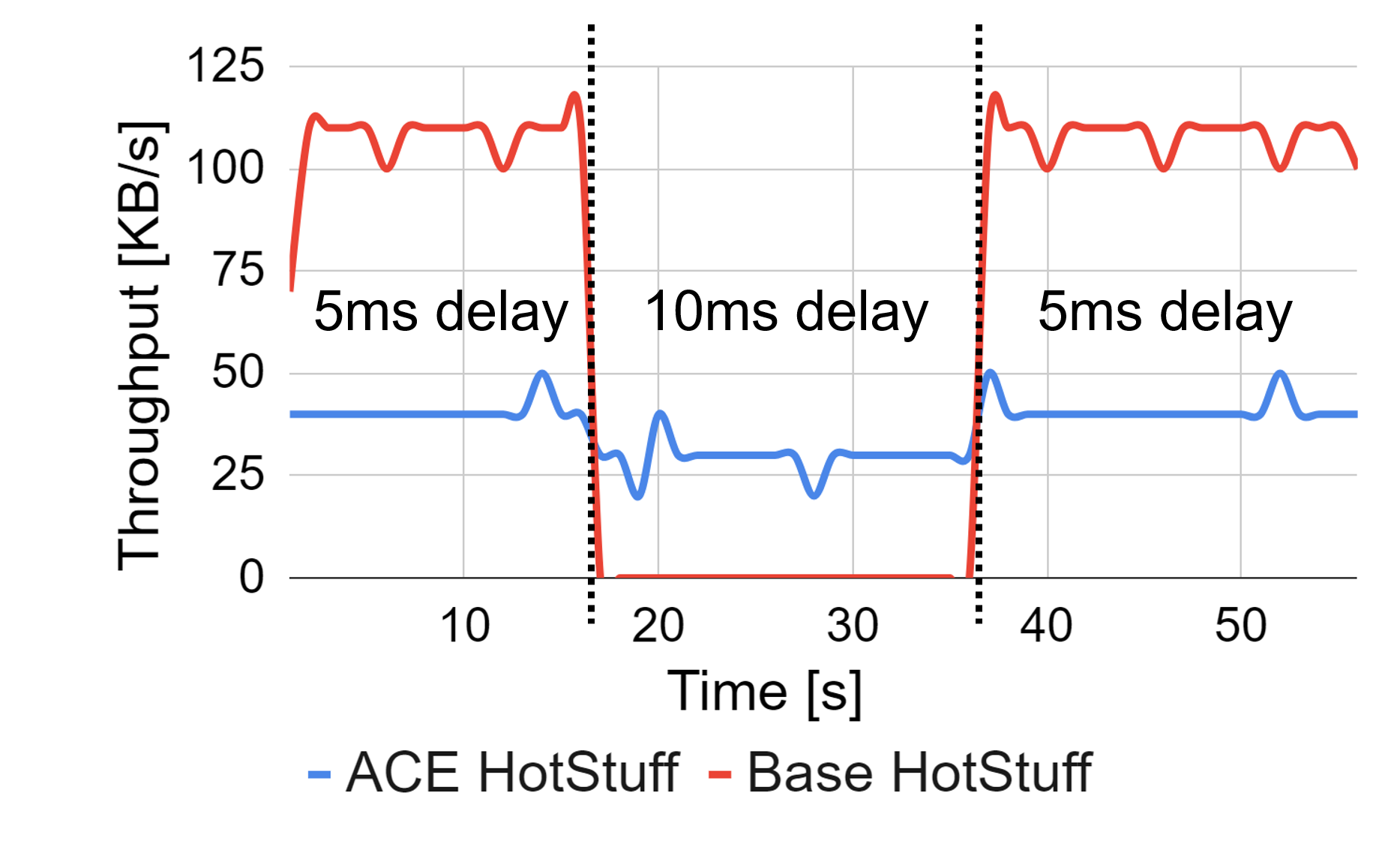}
    \caption{Throughput with a fluctuating
    transmission delay.}
    \label{fig:variabledelaythput}
\end{SCfigure}

\paragraph{Weak adaptive asynchrony attack.}

Note that since the views in base HotStuff are leader-based,
byzantine parties (or any other adversarial entity) can achieve
the same ``asynchronous'' effect presented above by only slowing down
the leaders.
In the next test we demonstrate the above using a
\emph{distributed denial of service (DDoS)} attack, in which leaders
are flooded with superfluous requests in an attempt to overload them
and delay their progress in the leader-based phase.


Figure~\ref{fig:DosAttackThpt} compares the throughput of ACE
HotStuff and base HotStuff, where the attack starts at the halfway
mark of the test.
The byzantine parties coordinate their
attack by adaptively choosing a single correct party and
flooding it with superfluous requests. In base HotStuff, byzantine
parties target correct leaders (byzantine leaders are making
progress).
In ACE HotStuff, there is no designated leader, therefore the
byzantine parties choose an arbitrary correct party to attack.
Our logs show that in base HotStuff progress is mainly
made in views where byzantines parties are leaders. If they would
not drive progress, the throughput would drop near $0$.

\sidecaptionvpos{figure}{c}
\begin{SCfigure}
      \includegraphics[width=0.28\textwidth]{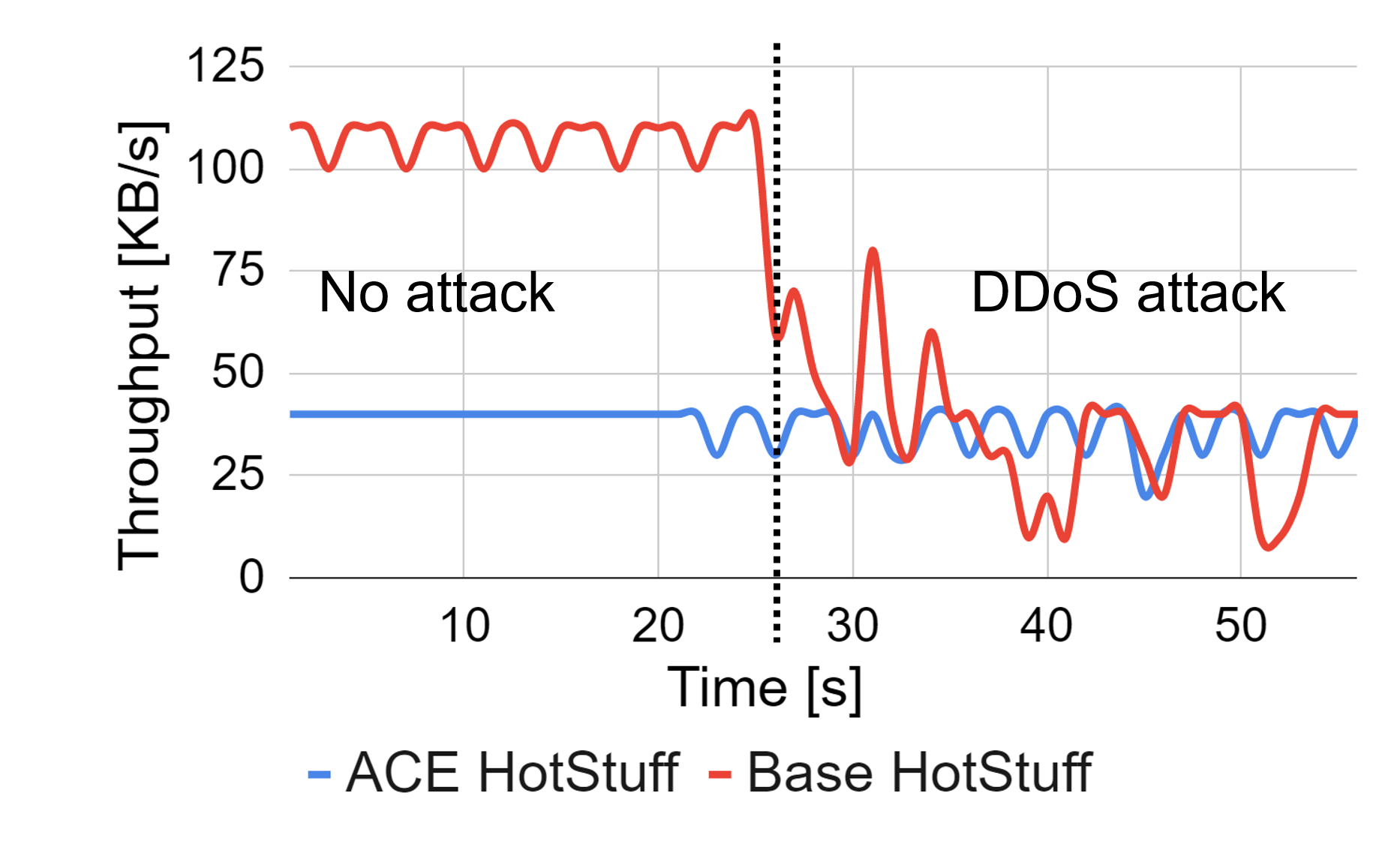}
    \caption{Throughput under DDoS attack.}
    \label{fig:DosAttackThpt}
\end{SCfigure}

\paragraph{Long conservative timeouts.}

The previous two scenarios operated base HotStuff with a fixed
aggressive view timer, which was based on the expected network delay.
This caused premature timer expiration during periods of increased
delays (due to asynchrony or attacks).
One might think that a possible solution can be to set a very long
timeouts that will never expire, thus letting the base HotStuff
protocol progress in network speed.
However, the downside of conservative timers is that
byzantine parties can perform a \emph{silent attack} on the
protocol's progress by not driving views when they are leaders,
forcing all parties to wait for the long timeouts to expire.

The next test evaluates base HotStuff with a conservative view timer
of $1$ second, fixed to be much higher than expected needed to commit
a view, under the silent attack starting at the half way mark.
Figure \ref{fig:StaticDelayAttackThpt} presents the results.
Before the attack, base HotStuff indeed progresses in network
speed, but during the attack, the throughput drops
significantly since a few consecutive byzantine leader might stall
progress for seconds.
In ACE HotStuff we see a much smaller drop, but more fluctuation.
This is due to the fact that byzantine leaders do not drive progress
in their LBV instances, and thus the expected number of waves until a
decision is now higher.


\sidecaptionvpos{figure}{c}
\begin{SCfigure}
      \includegraphics[width=0.28\textwidth]{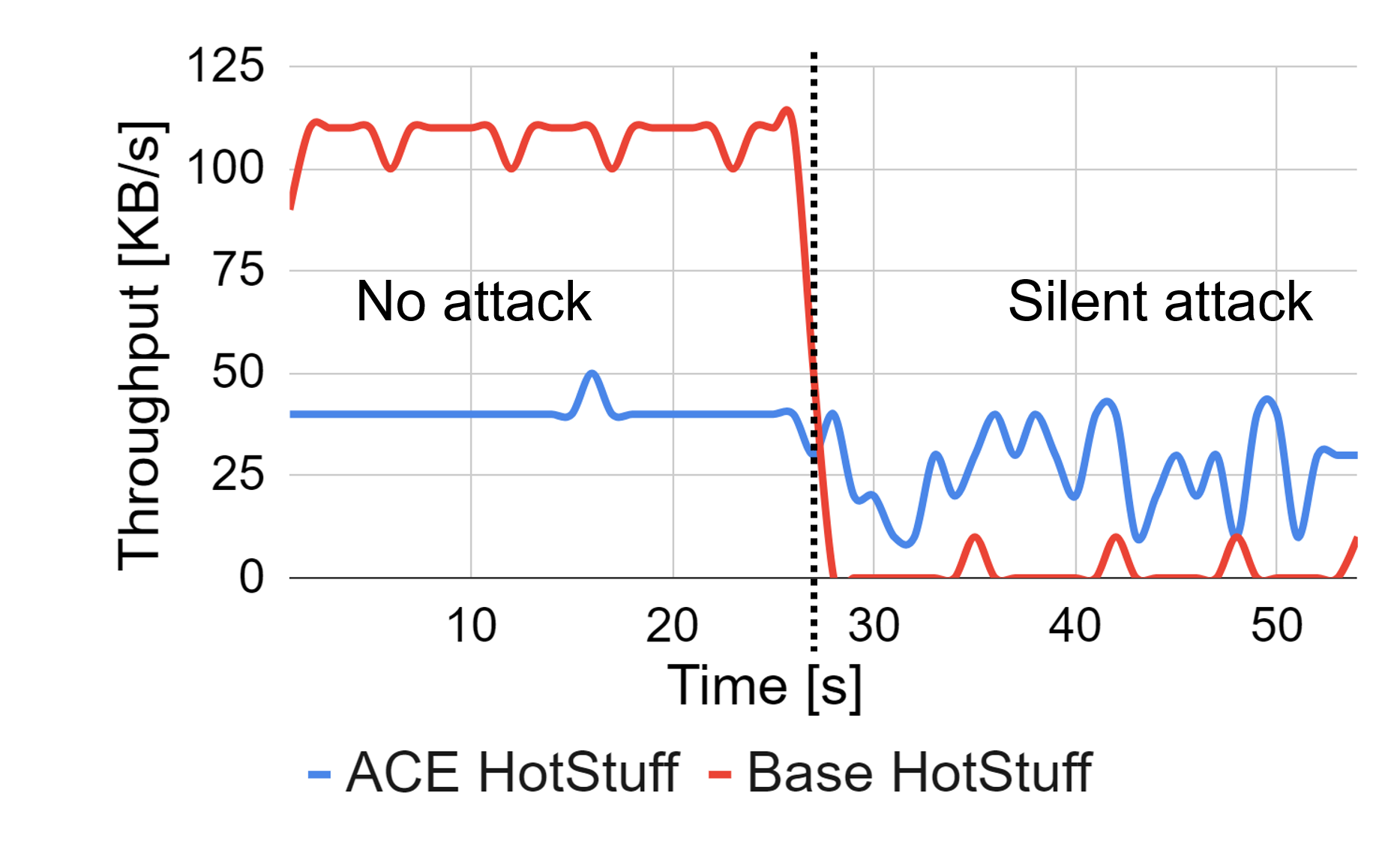}
    \caption{Throughput with conservative timeouts under
    byzantine silence attack.}
    \label{fig:StaticDelayAttackThpt}
\end{SCfigure}

\paragraph{Adjusting timeouts.}
As the scenarios above demonstrate, neither being too aggressive nor
being too conservative works well for base HotStuff during asynchrony
or attacks.
Therefore, in practice, when HotStuff is deployed it typically
adjusts timers during execution according to progress or lack of it.
The most common method (used also by PBFT~\cite{pbft} and
SBFT~\cite{sbft}) is to increase
timeouts whenever timers expires too early, and decrease them
whenever progress is made in order to try to learn the network delay
and adapt to it's dynamic changes.
%
%
To test this method, we implement an adaptive version,
starting with a delay of $t$. If a timeout is reached in a view
before a decision is made we set the next view's timeout to $1.25t$.
Otherwise, the next view's timeout is set to $0.8t$.

We evaluate this method against the following attack that combines
insights from the previous ones.
The results are shown in Figure \ref{fig:DynamicDelayAttackThpt}.
In the second half of the experiment, byzantine parties perform a
DDoS attack on correct leaders, causing the view timers to increase,
and then perform the silence attack (in views they
act as leaders) to stall progress as much as possible.
%
%
As expected, base HotStuff throughput drops to almost zero, whereas
ACE HotStuff continues driving decisions.
Same as in the previous test, ACE HotStuff suffers from fluctuation
due to the probability to choose a byzantine leader that did not made
progress in its LBV instance.
Another interesting phenomenon is the x2 performance drop of base
HotStuff before the attack begins compared to previous tests. 
This is due to the timeout adjustment mechanism, which reduces the
timers after every successful view, resulting in a too short timeout
in every second view.


\sidecaptionvpos{figure}{c}
\begin{SCfigure}
      \includegraphics[width=0.28\textwidth]{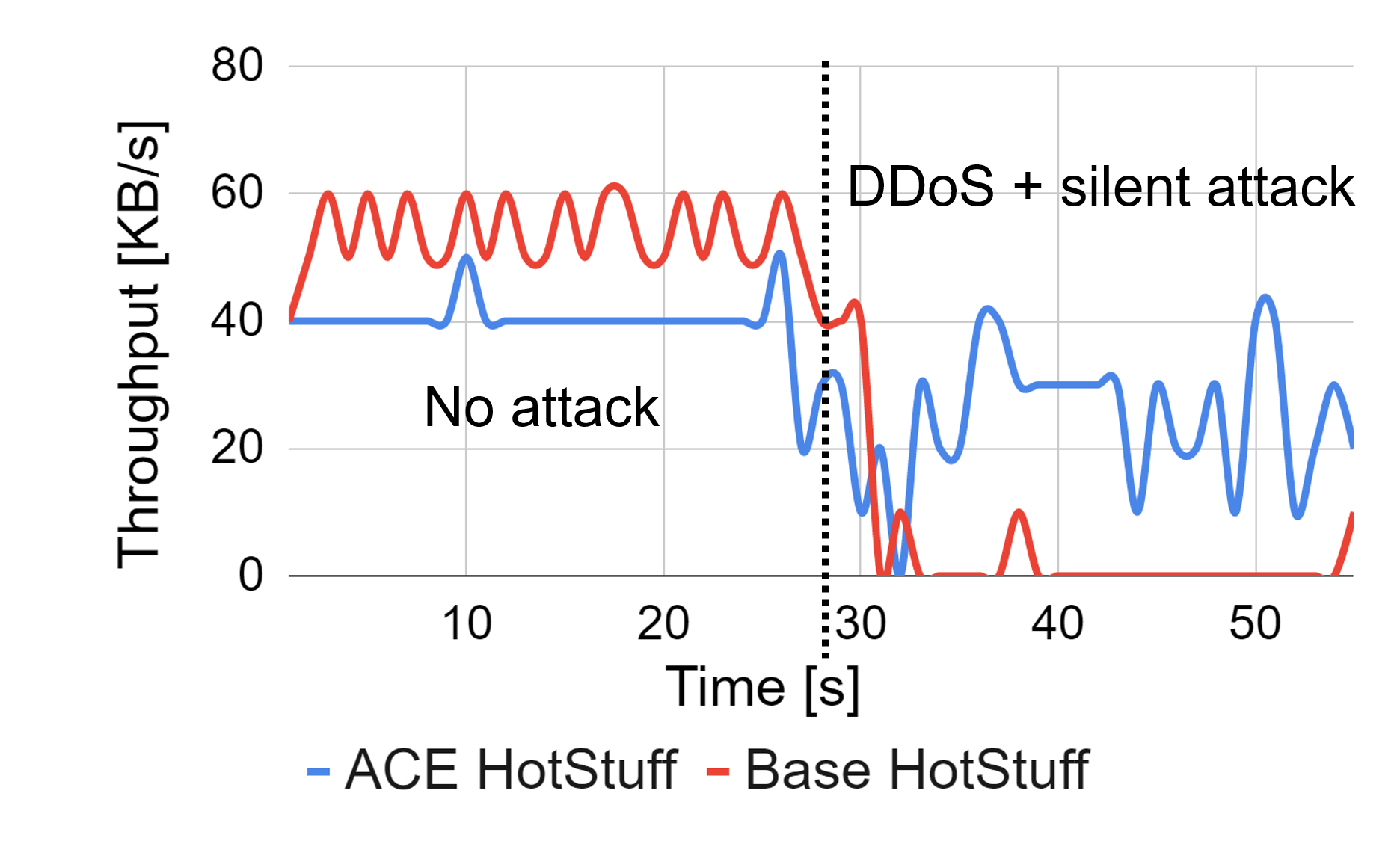}
    \caption{Throughput with adjusting timeouts under a combination of
    DDoS and silence attacks.}
    \label{fig:DynamicDelayAttackThpt}
\end{SCfigure}

While the timer adjustment algorithm can be further enhanced, it is
an arms race against the adversary -- for each method, there is an
adversarial response.
In addition, although this evaluation is focused on HotStuff, the
only ingredient of the algorithm that is under attack is the timeout,
hence the evaluation exemplifies the weakness of all leader-based
view by view algorithms.
Therefore, our evaluation suggests that the overhead of ACE in the
optimistic case is worth paying when high availability is desired
under all circumstances.

\section{Related work}
\label{sec:related-work}
The agreement problem was first introduced by Pease et
al.~\cite{pease1980reaching} almost 40 years ago, and has received an
enormous amount of attention since then~\cite{canetti1993fast, abd2005fault,
yin2003separating, pbft, zyzzyva,
amir2006scaling, martin2006fast, li2007beyond,
amir2007customizable, clement2009making, amir2011prime,
miller2016honey, liu2016xft, duan2014bchain, garay1998fully}.
One of the most important results is the
FLP~\cite{fischer1982impossibility} impossibility, proving that
deterministic solutions in the asynchronous communication models are
impossible.
Below we describe work that was done to circumvent the FLP
impossibility, present two related frameworks that were
previously proposed for the agreement problem, compare our SMR
definition to other systems in the literature, and discuss
alternative fairness definitions.

\paragraph{Agreement in the partial synchrony model.}
A practical approach to circumvent the FLP 
impossibility is to consider the partial synchrony
communication model~\cite{oki1988viewstamped,
ongaro2014search,sbft, hotstuff, zyzzyva}, which was first
proposed by Dwork et al.~\cite{dwork1988consensus} and later used by
seminal works like Paxos~\cite{lamport2001paxos} and
PBFT~\cite{pbft}.
As explained in detail in Section~\ref{sec:VbV}, protocols designed
for this model never violate safety, but provide progress only
during long enough synchronous periods.
Despite their limitations,
they are widely adopted in the
industry due to their relative simplicity compared to the
alternatives and their performance benefits during synchronous
periods.
For example, Casandra~\cite{cassandra},
Zookeeper~\cite{zookeeper}, and Google's
Spanner~\cite{spanner} implement a variant of
Paxos~\cite{lamport2001paxos}, and VMware's Concord~\cite{concord},
Facebook's Libra~\cite{libra} and IBM's
Hyperledger~\cite{hyperledger}, implement SBFT~\cite{sbft},
HotStuff~\cite{hotstuff} and PBFT~\cite{pbft}, respectively.

\paragraph{Agreement in the asynchrony model.}
As first shown by Ben-Or~\cite{ben1983another} and
Rabin~\cite{lehmann1981advantages}, the FLP impossibility result
does not stand randomization.
Meaning that the randomized version of the Agreement problem, which
guarantees termination with probability $1$, can be solved in the
asynchronous model provided that parties can flip random coins.
The algorithms in~\cite{ben1983another, lehmann1981advantages} are
very inefficient in terms of time and message complexity, and there
has been a huge effort to improve it over the years.
Some considered the theoretical full information model,
in which the adversary is computationally unbounded, and showed
more efficient algorithms that relax the failure
resilience threshold~\cite{kapron2010fast, king2013byzantine}.
These are beautiful theoretical results but too
complex to implement and maintain.

A more practical model for randomized asynchronous agreement is the
random oracle model in which the adversary is computationally bounded
and cryptographic assumptions (like the Decisional
Diffie–Hellman~\cite{diffie1976new}) are valid.
In the context of distributed computing, this model was first proposed
by Cachin et al.~\cite{Cachin2000RandomOI, CachinSecure}. In
\cite{CachinSecure} they proposed an almost optimal algorithm for the
agreement problem.
A variant of this algorithm was later implemented in
Honeybadger~\cite{miller2016honey} and Beat~\cite{duan2018beat},
which are the first academic asynchronous SMR systems.
The protocol in~\cite{CachinSecure} is optimal in terms of resilience
to failures and round complexity, but has an inefficient $O(n^3)$
communication cost.
Improving the communication cost was an open problem for almost 20 years,
until it was recently resolved in VABA~\cite{vaba}.
ACE borrows a lot from VABA~\cite{vaba}.
In fact, ACE can be seen as a generalization of the
approach introduced in VABA of letting $n$ parties progress in
parallel and then retrospectively choosing one. 

\paragraph{Frameworks for agreement.}

There are two previously proposed agreement
frameworks~\cite{guerraoui2010next, lamport2009vertical} that we are
aware of.

The next 700BFT~\cite{guerraoui2010next} framework proposes an
approach to compose different byzantine SMRs.
They observed that no byzantine SMR can outperform all others under
all circumstances, and introduce a general way for a system designer
to switch between implementations whenever the setting changes.
They defined \emph{Abstract}, which is an abortable SMR abstraction,
that captures the progress and safety requirements from a partially
synchronous SMR, and provides guidance on how multiple Abstract
instances should be composed.
Our work is very different from theirs.
While they defined an abstraction in order to compose different
SMR view-by-view implementations to achieve better performance in the
partially synchronous model, our LBV abstraction provides an API to
decouple the leader-based phase from the view-change phase in each
view, which in turn allows us to compose LBV instances in a novel way
that avoids leader demotions via timeouts and boost liveness in
asynchronous networks.

Vertical Paxos~\cite{lamport2009vertical} is a class of consensus
algorithms that separates the mechanism for reaching agreement from
the one that deals with failures.
The idea is to use a fast and small quorum of parties to drive
agreement, and have an auxiliary reconfiguration master to
reconfigure this quorum whenever progress stalls.
The protocol for agreement relies on the participation of all
parties in the dedicated quorum, and thus stalls
whenever some party fails.
The master is emulated by a bigger quorum, which uses an agreement
protocol to agree on reconfiguration, and thus can tolerate
failures.


\paragraph{State machine replication.}
Paxos~\cite{lamport2001paxos} (crash-failure model) and
PBFT~\cite{pbft} (byzantine model) were the first to show how to build
an SMR from a single-shot
agreement problem.
In both cases, similarly to ACE, parties use a single-shot
agreement instance to agree on the value of every slot, but contrary
to our algorithm, they do not satisfy the FIFO property since they do
not make sure all parties learn the decision value before moving to
the next slot.
For some applications, e.g., Blockchains, this can be crucial since
the validity of a value sometimes depends on decision values of
previous slots~\cite{bitcoin}.

Moreover, for practical reasons, systems implementing the Paxos and
PBFT algorithms use
periodic checkpoints~\cite{concord, libra, bessani2014state}, in
which parties
exchange all the decision values made since the last checkpoint in
order to free resources associated with these slots.
The cost of these checkpoints is quadratic in the number of
decision values and since better than quadratic communication per
decision is impossible~\cite{dolev1983authenticated} we decide to
avoid these checkpoints.
Instead, we formally define the FIFO and Strong halting properties
and perform a quadratic forwarding mechanism after each slot.
It is important to note that this is a design choice; ACE's
abstractions can be used in a similar way to build SMR with different
guarantees.


%

\paragraph{Fairness.}

Although the Agreement and SMR problems have been studied for
many years, the question of fairness therein was only recently asked,
and we are aware of only few solutions that provide some notion of
it~\cite{amir2011prime, miller2016honey, lev2019fairledger, vaba}.
Prime~\cite{amir2011prime} extends PBFT~\cite{pbft} to
guarantee that values are committed in a bounded number of
slots after they first proposed, and
FairLedger~\cite{lev2019fairledger} uses batching to ensures that
all correct party commits a value in every batch.
However, in contrast to ACE, both protocols are able to guarantee
fairness only during synchronous periods.
Honeybadger~\cite{miller2016honey} is an asynchronous
protocol that, similarly to FairLedger, batches values proposed by
different parties and commits them together atomically.
It probabilistically bounds the number of epochs
(and accordingly the number of slots) until a value is committed,
after being submitted to $n-f$ parties.
The VABA~\cite{vaba} protocol does not use batching, and provides a
per slot guarantee that bounds the probability to choose a
value proposed by a correct party during asynchronous periods.
ACE provides similar fairness guarantees during asynchrony, but
also guarantees equal chance for each correct party during synchrony.

\section{Discussion}
\label{sec:discussion}

In this paper we introduced ACE: a general model agnostic framework
for boosting asynchronous liveness of any leader-based SMR system
designed for the partially synchronous model.
The main ingredient is the novel \emph{LBV} abstraction that
encapsulates the properties of a single view in leader-based
view-by-view algorithms, while providing an API to control the
scheduler of the two phases, leader-based and view-change, in each
view.
Exploiting this separation, ACE provides a novel algorithm that
composes LBV instances in a way that avoids timers
and provides a randomized asynchronous SMR solution.

ACE is model agnostic, meaning that it does not add any assumptions on
top of what are assumed in the instantiated LBV implementation, thus
provides a generic liveness boosting for both byzantine and
crash-failure SMRs. 
In order to instantiate ACE with a specific SMR
algorithm, all a system designer needs to do is alter the code
of a single view to support LBV's API; this should not be too
complicated as the view logic must already implicitly satisfy the
required API's properties.

In addition to boosting liveness, ACE is designed in a way that
inherently provides fairness due to its randomized election of
leaders in retrospect.
Moreover, ACE provides a clear separation between safety, which
relies on the LBV implementation, and liveness, which is given by the
framework.
As a result, a system designer that chooses to instantiate ACE gets a
modular SMR implementation that is easier to prove correct and
maintain -- if a better agreement protocol is published, all the
designer needs to do in order to integrate it in the system is to
alter the LBV implementation accordingly.

To demonstrate the power of ACE we implemented it, instantiated it
with the state of the art HotStuff~\cite{hotstuff} protocol, and
compared its performance to the base HotStuff implementation.
Our results show that while ACE suffers a 2x performance degradation
in the optimistic, synchronous, failure-free case, it enjoys absolute
superiority during asynchronous periods and network attacks.


\begin{acks}                            
We thank Ittai Abraham for helpful initial discussions and Dahlia
Malkhi for reviewing drafts and suggesting valuable
improvements.
We are also grateful to Guy Golan-Gueta, Guy Goren, Idit Keidar,
Eleftherios Kokoris-Kogias, and David Tennenhouse for their useful
feedback.

\end{acks}

\bibliography{bibliography}

%

\end{document}